\documentclass[11pt,letterpaper,american]{article}
\usepackage{floatrow}
\usepackage{amsmath}
\usepackage{amsthm}
\usepackage{amssymb}
\usepackage{amsfonts}
\usepackage{color}
\usepackage[hidelinks,bookmarksnumbered]{hyperref}
\usepackage{algorithm}
\usepackage{caption}
\usepackage{MnSymbol}
\usepackage{url}
\usepackage[margin=1in]{geometry}
\usepackage[labelfont=bf,font=small,margin=1em]{caption} 
\usepackage{mathtools}
\usepackage{refcount}
\usepackage{pbox}
\usepackage{colonequals}
\usepackage{wref}
\usepackage{enumitem}

\newcommand\ce{\coloneq}

\usepackage{tikz}
\usetikzlibrary{shapes,decorations.pathreplacing}
\usetikzlibrary{arrows.meta,automata,positioning,arrows}

\newtheoremstyle{sl}
  {1.25\medskipamount}
  {1.25\medskipamount}
  {\slshape}
  {}
  {\bfseries}
  {.}
  {.5em}
  {\thmname{#1}\thmnumber{ #2}\thmnote{ (#3)}}
\theoremstyle{sl}

\newtheorem{theorem}{Theorem}
\newtheorem{lemma}[theorem]{Lemma}
\newtheorem{corollary}[theorem]{Corollary}

\newtheorem{definition}[theorem]{Definition}

\newtheoremstyle{defn}
  {1.\medskipamount}
  {1.\medskipamount}
  {\normalfont}
  {}
  {\bfseries}
  {.}
  {.5em}
  {\thmname{#1}\thmnumber{ #2}\thmnote{ (#3)}}
\theoremstyle{defn}

\newtheorem{remark}[theorem]{Remark}

\newlength\tmp
\newcommand\twopart[2]{%
	\settowidth\tmp{#1}%
	#1\mbox{\parbox[t]{\linewidth - \tmp}{\raggedright{}#2}}%
}

\newcommand\ie{i.\kern.1eme.}
\newcommand\eg{e.\kern.1emg.}

\usepackage{fancyhdr}
\renewcommand\subsectionmark[1]{}

\title{\bfseries Lazy Search Trees}
\author{Bryce Sandlund$^1$ and Sebastian Wild$^2$}
\date{$^1$David R. Cheriton School of Computer Science, University of Waterloo\\
         $^2$Department of Computer Science, University of Liverpool\\[2ex]
\small\today}

\usepackage{microtype}

\begin{document}
\maketitle
\thispagestyle{empty}

\begin{abstract}\sloppy
We introduce the lazy search tree data structure. The lazy search tree is a comparison-based data structure on the pointer machine that supports order-based operations such as rank, select, membership, predecessor, successor, minimum, and maximum while providing dynamic operations insert, delete, change-key, split, and merge. We analyze the performance of our data structure based on a partition of current elements into a set of \emph{gaps} $\{\Delta_i\}$ based on rank. A query falls into a particular gap and $\emph{splits}$ the gap into two new gaps at a rank $r$ associated with the query operation. If we define $B = \sum_i |\Delta_i| \log_2(n/|\Delta_i|)$, our performance over a sequence of $n$ insertions and $q$ distinct queries is $O(B + \min(n \log \log n, n \log q))$. We show $B$ is a lower bound.

Effectively, we reduce the insertion time of binary search trees from $\Theta(\log n)$ to $O(\min(\log(n/|\Delta_i|) + \log \log |\Delta_i|, \; \log q))$, where $\Delta_i$ is the gap in which the inserted element falls. Over a sequence of $n$ insertions and $q$ queries, a time bound of $O(n \log q + q \log n)$ holds; better bounds are possible when queries are non-uniformly distributed. As an extreme case of non-uniformity, if all queries are for the minimum element, the lazy search tree performs as a priority queue with $O(\log \log n)$ time insert and decrease-key operations. The same data structure supports queries for \emph{any} rank, interpolating between binary search trees and efficient priority queues.

Lazy search trees can be implemented to operate mostly on arrays, requiring only $O(\min(q, n))$ pointers, suggesting smaller memory footprint, better constant factors, and better cache performance compared to many existing efficient priority queues or binary search trees.
Via direct reduction, our data structure also supports the efficient access theorems of the splay tree, providing a powerful data structure for non-uniform element access, both when the number of accesses is small and large.

\end{abstract}

\clearpage
\pagenumbering{arabic}

\section{Introduction}

We consider data structures supporting order-based operations such as rank, select, membership, predecessor, successor, minimum, and maximum while providing dynamic operations insert, delete, change-key, split, and merge. The classic solution is the binary search tree (BST), perhaps the most fundamental data structure in computer science. The original unbalanced structure is credited to papers by Booth and Colin~\cite{Booth60}, Douglas~\cite{Douglas59}, Windley~\cite{Windley60}, and Hibbard~\cite{Hibbard62} in the early 1960's. Since then, a plethora of \emph{balanced} binary search tree data structures have been proposed~\cite{Adelson62,Bayer72,Bayer72b,Andersson89,Galperin93,Sleator85,Seidel96,NIEVER72}, notable examples including AVL trees~\cite{Adelson62}, red-black trees~\cite{Bayer72}, and splay trees~\cite{Sleator85}. A balanced binary search tree is a staple data structure included in nearly all major programming language's standard libraries and nearly every undergraduate computer science curriculum. The data structure is the dynamic equivalent of binary search in an array, allowing searches to be performed on a changing set of keys at nearly the same cost. Extending to multiple dimensions, the binary search tree is the base data structure on which range trees~\cite{Bentley79}, segment trees~\cite{Bentley77}, interval trees~\cite{Edelsbrunner80,McCreight80}, $k$d-trees~\cite{Bentley75}, and priority search trees~\cite{McCreight85} are all built.

The theory community has long focused on developing binary search trees with efficient \emph{query} times. Although $\Omega(\log n)$ is the worst-case time complexity of a query, on non-uniform access sequences binary search trees can perform better than logarithmic time per query by, for example, storing recently accessed elements closer to the root. The splay tree was devised as a particularly powerful data structure for this purpose~\cite{Sleator85}, achieving desirable access theorems such as static optimality, working set, scanning theorem, static finger, and dynamic finger~\cite{Sleator85,Cole2000a,Cole2000b}. The most famous performance statement about the splay tree, however, is the unproven dynamic optimality conjecture, which claims that the performance of the splay tree is within a constant factor of any binary search tree on any sufficiently long access sequence, subsuming all other access theorems. Proving this conjecture is widely considered one of the most important open problems in theoretical computer science, receiving vast attention by data structure researchers~\cite{Allen78,Demaine09,Demaine07,Iacono16,Bose2020,Sleator83,Wilber89,Kozma2019,Chalermsook15,Iacono01,Badoiu06}. Despite ultimately remaining unsolved for nearly four decades, this topic continues to receive extensive treatment~\cite{Iacono16,Bose2020,Chalermsook15,Levy19,Badoiu06}. %

Although widely considered for the task in literature, the binary search tree is not the most efficient data structure for the standard dictionary abstract data type: in practice, dictionaries are almost always implemented by hash tables, which support $O(1)$ time insert, delete, and look-up in expectation~\cite{Fredman84,Pagh04}. The advantage of binary search trees, over hash tables, is that they support \emph{order-based} operations. We call dictionaries of this type \emph{sorted dictionaries}, to differentiate them from the simpler data structures supporting only membership queries.

If we limit the order-based operations required of our sorted dictionary to queries for the minimum or maximum element (or both), a number of alternative solutions to the binary search tree have been developed, known as priority queues. The first of which was the binary heap, invented in 1964 for the heapsort algorithm~\cite{Williams64}. The binary heap achieves asymptotic complexity equivalent to a binary search tree, though due to the storage of data in an array and fast average-case complexity, it is typically the most efficient priority queue in practice. Later, the invention of the binomial heap showed that the merging of two arbitrary priority queues could be supported efficiently~\cite{Vuillemin78,Brown78}, thus proving that the smaller operation set of a priority queue allows more efficient runtimes. The extent to which priority queues can outperform binary search trees was fully realized with the invention of Fibonacci heaps, which showed insertion, merge, and an additional decrease-key operation can all be supported in $O(1)$ amortized time~\cite{Fredman87}. Since then, a number of priority queues with running times close to or matching Fibonacci heaps have been developed~\cite{Fredman86,Chan09,Brodal12,Elmasry09,Haeupler11,Brodal96,Hansen15}. We refer to such priority queues with $o(\log n)$ insertion and decrease-key costs as \textit{efficient} priority queues, to distinguish them from their predecessors and typically simpler counterparts with $O(\log n)$ insertion and/or decrease-key cost.

The history of efficient priority queues contrasts that of binary search trees. Efficient priority queues have been developed for the case when the number of queries is significantly less than the number of insertions or updates. On the other hand, research on binary search trees has focused on long sequences of element \emph{access}. %
Indeed, the dynamic optimality conjecture starts with the assumption that $n$ elements are already present in the binary search tree, placing any performance improvements by considering insertion cost entirely outside of the model. However, the theory of efficient priority queues shows that on some operation sequences, the efficiency gains due to considering insertion cost can be as much as a $\Theta(\log n)$ factor, showing an as-of-yet untapped area of potential optimization for data structures supporting the operations of a binary search tree. Further, aside from the theoretically-appealing possibility of the unification of the theories of efficient priority queues and binary search trees, the practicality of improved insertion performance is arguably greater than that of improved access times. 
For the purpose of maintaining keys in a database, for example, an insert-efficient data structure can provide superior runtimes when the number of insertions dominates the number of queries, a scenario that is certainly the case for some applications~\cite{ONeil96,Brodal03} and is, perhaps, more likely in general.
Yet in spite of these observations, almost no research has been conducted that seriously explores this frontier~\cite{Bose13}.

We attempt to bridge this gap. We seek a general theory of comparison-based sorted dictionaries that encompasses efficient priority queues and binary search trees, providing the operational flexibility of the latter with the efficiency of the former, when possible.
We do not restrict ourselves to any particular BST or heap model; while these models with their stronger lower bounds are theoretically informative, for the algorithm designer these lower bounds in artificially constrained models are merely indications of what \emph{not} to try. If we believe in the long-term goal of improving algorithms and data structures in practice~-- an objective we think will be shared by the theoretical computer science community at large~-- we must also seek the comparison with lower bounds in a more permissive model of computation.%

We present \emph{lazy search trees}. The lazy search tree is the first data structure to support the general operations of a binary search tree while providing superior insertion time when permitted by query distribution. We show that the theory of efficient priority queues can be generalized to support queries for any rank, via a connection with the multiple selection problem. Instead of sorting elements upon insertion, as does a binary search tree, the lazy search
delays sorting to be completed incrementally while queries are answered. %
 A binary search tree and an efficient priority queue are special cases of our data structure that result when queries are frequent and uniformly distributed or only for the minimum or maximum element, respectively. While previous work has considered binary search trees in a ``lazy" setting (known as ``deferred data structures'')~\cite{Karp88,CHING90} and multiple selection in a dynamic setting~\cite{Barbay15,BARBAY16}, no existing attempts fully distinguish between insertion and query operations, severely limiting the generality of their approaches. The model we consider gives all existing results as corollaries, unifying several research directions and providing more efficient runtimes in many cases, all with the use of a single data structure.
Before we can precisely state our results, we must formalize the model in which they are attained.

\subsection{Model and Results}
\label{sec:model}

We consider comparison-based data structures on the pointer machine. While we suggest the use of arrays in the implementation of our data structure in practice, constant time array access is not needed for our results. Limiting operations to a pointer machine has been seen as an important property in the study of efficient priority queues, particularly in the invention of strict Fibonacci heaps~\cite{Brodal12} compared to an earlier data structure with the same worst-case time complexities~\cite{Brodal96}.

We consider data structures supporting the following operations on a dynamic multiset $S$ with 
(current) size $n = |S|$. We call such data structures \emph{sorted dictionaries}:
\begin{itemize}
	\item \texttt{Construction($S$)} $\ce$ Construct a sorted dictionary on the set $S$.
	\item \texttt{Insert($e$)} $\ce$ Add element $e=(k,v)$ to $S$, using key $k$ for comparisons; 
		(this increments $n$).
	\item \texttt{RankBasedQuery($r$)} $\ce$ Perform a rank-based query pertaining to rank $r$ on $S$.
	\item \texttt{Delete(ptr)} $\ce$ Delete the element pointed to by \texttt{ptr} from $S$; (this decrements $n$).
	\item \texttt{ChangeKey(ptr,\,$k'$)} $\ce$ Change the key of the element pointed to by \texttt{ptr} to $k'$.
	\item \twopart{\texttt{Split($r$)} $\ce$ }{Split $S$ at rank $r$, returning two sorted dictionaries $T_1$ and $T_2$ of $r$ and $n-r$ elements, respectively, such that for all $x \in T_1$, $y \in T_2$, $x \leq y$.}
	\item \twopart{\texttt{Merge($T_1$,\,$T_2$)} $\ce$ }{Merge sorted dictionaries $T_1$ and $T_2$ and return the result, given that for all $x \in T_1$, $y \in T_2$, $x \leq y$.}
\end{itemize}

We formalize what queries are possible within the stated operation \texttt{RankBasedQuery($r$)} in \wref{sec:prelim}. For now, we informally define a rank-based query as any query computable in $O(\log n)$ time on a (possibly augmented) binary search tree and in $O(n)$ time on an unsorted array. Operations rank, select, contains, successor, predecessor, minimum, and maximum fit our definition. To each operation, we associate a rank $r$: for membership and rank queries, $r$ is the rank of the queried element (in the case of duplicate elements, an implementation can break ties arbitrarily), and for select, successor, and predecessor queries, $r$ is the rank of the element returned; minimum and maximum queries are special cases of select with $r = 1$ and $r = n$, respectively.

The idea of lazy search trees is to maintain a partition of current elements in the data structure into what we will call \emph{gaps}. We maintain a set of $m$ gaps $\{\Delta_i\}$, $1 \leq i \leq m$, where a gap $\Delta_i$ contains a bag of elements. Gaps satisfy a total order, so that for any elements $x \in \Delta_i$ and $y \in \Delta_{i+1}$, $x \leq y$. Internally, we will maintain structure within a gap, but the interface of the data structure and the complexity of the operations is based on the distribution of elements into gaps, assuming nothing about the order of elements within a gap. Intuitively, binary search trees fit into our framework by restricting $|\Delta_i| = 1$, so each element is in a gap of its own, and we will see that priority queues correspond to a single gap $\Delta_1$ which contains all elements. Multiple selection corresponds to gaps where each selected rank marks a separation between adjacent gaps.

To insert an element $e = (k, v)$, where $k$ is its key and $v$ its value, we find a gap $\Delta_i$ in which it belongs without violating the total order of gaps (if $x \leq k$ for all $x \in \Delta_i$ and $k \leq y$ for all $y \in \Delta_{i+1}$, we may place $e$ in either $\Delta_i$ or $\Delta_{i+1}$; implementations can make either choice). Deletions remove an element from a gap; if the gap is now empty we can remove the gap. When we perform a query, we first narrow the search down to the gap $\Delta_i$ in which the query rank $r$ falls (formally, $\sum_{j=1}^{i-1} |\Delta_j| < r \leq \sum_{j=1}^i |\Delta_j|$). We then answer the query using the elements of $\Delta_i$ and \emph{restructure} the gaps in the process. 
We split gap $\Delta_i$ into two gaps $\Delta'_i$ and $\Delta'_{i+1}$ such that the total order on gaps is satisfied and the rank $r$ element is either the largest in gap $\Delta'_i$ or the smallest in gap $\Delta'_{i+1}$; specifically, either $|\Delta'_i| + \sum_{j=1}^{i-1} |\Delta_j| = r$ or $|\Delta'_i| + \sum_{j=1}^{i-1} |\Delta_j| = r-1$. (Again, implementations can take either choice. We will assume a maximum query to take the latter choice and all other queries the former. More on the choice of $r$ for a given query is discussed in \wref{sec:prelim}. Our analysis will assume two new gaps replace a former gap as a result of each query. Duplicate queries or queries that fall in a gap of size one follow similarly, in $O(\log n)$ time.) We allow duplicate insertions.

Our performance theorem is the following.

\begin{theorem}[Lazy search tree runtimes]
\label{thm:main}
Let $n$ be the total number of elements currently in the data structure and let $\{\Delta_i\}$ be defined as above (thus $\sum_{i=1}^m |\Delta_i| = n$). Let $q$ denote the total number of queries. Lazy search trees support the operations of a sorted dictionary on a dynamic set $S$ in the following runtimes:
\begin{itemize}
	\item {\normalfont \texttt{Construction($S$)}} 
		in $O(n)$ worst-case time, where $|S| = n$.
	\item {\normalfont \texttt{Insert($e$)}} 
		in $O(\min(\log(n/|\Delta_i|) + \log \log |\Delta_i|,\: \log q))$ worst-case time%
		\footnote{\label{fn:note1}
			To simplify formulas, we distinguish between $\log_2(x)$, the binary logarithm for any $x > 0$, 
			and $\log(x)$, which we define as $\max(\log_2(x),1)$.
		}%
		, where $e = (k,v)$ is such that $k \in \Delta_i$.
	\item {\normalfont \texttt{RankBasedQuery($r$)}} 
		in $O(x \log c + \log n)$ amortized time, where the larger resulting gap from the split is of size $cx$ and the other  gap is of size $x$.
	\item {\normalfont \texttt{Delete(ptr)}} 
		in $O(\log n)$ worst-case time.
	\item {\normalfont \texttt{ChangeKey(ptr,\,$k'$)}} 
		in $O(\min(\log q,\, \log \log |\Delta_i|))$ worst-case time, where the element pointed to by $\mathtt{ptr}$, $e=(k,v)$, has $k \in \Delta_i$ and $k'$ moves $e$ closer to its closest query rank%
		\footnote{\label{fn:note2}
			The closest query rank of $e$ is the closest boundary of $\Delta_i$ that was created in response to a query. For gaps $\Delta_i$ with $1\neq i \neq m$, this is the boundary of $\Delta_i$ that is closer with respect to the rank of $k$. Gaps $\Delta_1$ and $\Delta_m$ may follow similarly to $i \neq 1, m$ if a minimum or maximum has been extracted. With a single gap $\Delta_1$, increase-key is supported efficiently if maximums have been removed and decrease-key is supported efficiently if minimums have been removed. If both have been removed, the gap functions as in the general case for $i \neq 1, m$. Intuitively, this is configured to support the behavior of decrease-key/increase-key without special casing when the data structure is used as a min-heap/max-heap.

		} in $\Delta_i$;
		otherwise, {\normalfont \texttt{ChangeKey(ptr,\,$k'$)}} takes $O(\log n)$ worst-case time.
	\item {\normalfont \texttt{Split($r$)}} 
		in time according to {\normalfont \texttt{RankBasedQuery($r$)}}.
	\item {\normalfont \texttt{Merge($T_1$,\,$T_2$)}} 
		in $O(\log n)$ worst-case time.
\end{itemize}
Define $B = \sum_{i=1}^m |\Delta_i| \log_2(n/|\Delta_i|)$. Then over a series of insertions and queries with no duplicate queries, the total complexity is $O(B + \min(n \log \log n, n \log q))$.
\end{theorem}

We can also bound the number of pointers needed in the data structure.

\begin{theorem}[Pointers]
	\label{thm:qbounds}
An array-based lazy search tree implementation requires $O(\min(q, n))$ pointers.
\end{theorem}

By reducing multiple selection to the sorted dictionary problem, we can show the following lower bound.

\begin{theorem}[Lower bound]
\label{thm:lb}
Suppose we process a sequence of operations resulting in gaps $\{\Delta_i\}$. Again define $B = \sum_{i=1}^m |\Delta_i| \log_2(n/|\Delta_i|)$. Then this sequence of operations requires $B-O(n)$ comparisons and $\Omega(B + n)$ time in the worst case.
\end{theorem}

\wref{thm:lb} indicates that lazy search trees are at most an additive $O(\min(n \log \log n, n \log q))$ term from optimality over a series of insertions and distinct queries. This gives a lower bound on the per-operation complexity of {\normalfont \texttt{RankBasedQuery($r$)} of $\Omega(x \log c)$; the bound can be extended to $\Omega(x \log c + \log n)$ if we amortize the total work required of splitting gaps to each individual query operation. A lower bound of $\Omega(\min(\log(n/|\Delta_i|), \log m))$ can be established on insertion complexity via information theory. We describe all lower bounds in \wref{sec:bounds}.

We give specific examples of how lazy search trees can be used and how to analyze its complexity according to \wref{thm:main} in the following subsection.

\subsection{Example Scenarios}
\label{examples}

Below, we give examples of how the performance of \wref{thm:main} is realized in different operation sequences. While tailor-made data structures for many of these applications are available, lazy search trees provide a \emph{single} data structure that seamlessly adapts to the actual usage pattern while achieving optimal or near-optimal performance for all scenarios in a uniform way.
\begin{enumerate}
	\item \textbf{Few Queries:} The bound $B = \sum_{i=1}^{m} |\Delta_i| \log_2 (n/|\Delta_i|)$ satisfies $B = O(n \log q + q \log n)$. In the worst case, queries are uniformly distributed, and the lower bound $B = \Theta(n \log q + q \log n)$. Over a sequence of insertions and queries without duplicate queries, our performance is optimal $O(n \log q + q \log n)$. If $q = n^\epsilon$ for constant $\epsilon > 0$, lazy search trees improve upon binary search trees by a factor $1/\epsilon$. If $q = O(\log^c n)$ for some $c$, lazy search trees serve the operation sequence in $O(cn \log \log n)$ time and if $q = O(1)$, lazy search trees serve the operation sequence in linear time. Although it is not very difficult to modify a previous ``deferred data structure" to answer a sequence of $n$ insertions and $q$ queries in $O(n \log q + q \log n)$ time (see \wref{sec:deferred-data-structures}), 
	to the best of our knowledge, such a result has not appeared in the literature.
	
	\item \textbf{Clustered Queries:} Suppose the operation sequence consists of $q/k$ ``range queries'', each requesting $k$ consecutive keys, with interspersed insertions following a uniform distribution. Here, $B = O(n \log (q/k) + q \log n)$, where $q$ is the total number of keys requested.
	If the queried ranges are uniformly distributed, $B = \Theta(n \log (q/k) + q \log n)$, with better results possible if the range queries are non-uniform. Our performance on this operation sequence is $O(B + \min(n \log \log n, n \log q))$, tight with the lower bound if $k = \Theta(1)$ or $q/k = \Omega(\log n)$. Similarly to Scenario 1, we pay $O(n \log (q/k))$ in total for the first query of each batch; however, each successive query in a batch costs only $O(\log n)$ time as the smaller resulting gap of the query contains only a single element. We will see in \wref{sec:bounds} that we must indeed pay $\Omega(\log n)$ amortized time per query in the worst case; again our advantage is to reduce insertion costs.
	Note that if an element is inserted within the elements of a previously queried batch, these insertions take $O(\log n)$ time. However, assuming a uniform distribution of element insertion throughout, this occurs on only an $O(q/n)$ fraction of insertions in expectation, at total cost $O(n \cdot q/n \cdot \log n) = O(q \log n)$.
	Other insertions only need an overall $O(n \log (q/k) + \min(n \log \log n, n \log q))$ time.

	\item \textbf{Selectable Priority Queue:} If every query is for a minimum element, each query takes $O(\log n)$ time and separates the minimum element into its own gap and all other elements into another single gap. Removal of the minimum destroys the gap containing the minimum element, leaving the data structure with a single gap $\Delta_1$. All inserted elements fall into this single gap, implying insertions take $O(\log \log n)$ time. Further, the \texttt{ChangeKey(ptr,\,$k'$)} operation supports decrease-key in $O(\log \log n)$ time, since all queries (and thus the closest query) are for rank $1$. Queries for other ranks are also supported, though if queried, these ranks are introduced into the analysis, creating more gaps and potentially slowing future insertion and decrease-key operations, though speeding up future selections. The cost of a selection is $O(x \log c + \log n)$ amortized time, where $x$ is the distance from the rank selected to the nearest gap boundary (which was created at the rank of a previous selection) and $c = |\Delta_i|/x - 1$, where the selection falls in gap $\Delta_i$. If no selections have previously occurred, $x$ is the smaller of the rank or $n$ minus the rank selected and $c = n/x - 1$.

	Interestingly, finding the $k$th smallest element in a binary min-heap can be done in $O(k)$ time~\cite{Frederickson93}, yet we claim our runtime optimal! The reason is that neither runtime dominates in an amortized sense over the course of $n$ insertions. Our lower bound indicates that $\Omega(B + n)$ time must be taken over the course of multiple selections on $n$ elements in the worst case. %
	In Frederickson's algorithm, the speed is achievable because a binary heap is more structured than an unprocessed set of $n$ elements and only a single selection is performed; the ability to perform further selection on the resulting pieces is not supported.
	On close examination, lazy search trees can be made to answer the selection query alone without creating additional gaps in $O(x + \log n)$ amortized time or only $O(x)$ time given a pointer to the gap in which the query falls (such modification requires fulfilling \wref[Rules]{rule:merge} and~\ref{inv:credits} on category $A$ intervals in \wref{sec:queryanalysis}).

	\item \textbf{Double-Ended Priority Queue:} If every query is for the minimum or maximum element, again each query takes $O(\log n)$ time and will separate either the minimum or maximum element into its own gap and all other elements into another single gap. The new gap is destroyed when the minimum or maximum is extracted. As there is only one gap $\Delta_1$ when insertions occur, insertions take $O(\log \log n)$ time. In this case, our data structure natively supports an $O(\log \log n)$ time decrease-key operation for keys of rank $n/2$ or less and an $O(\log \log n)$ time increase-key operation for keys of rank greater than $n/2$. Further flexibility of the change-key operation is discussed in \wref{sec:changekeyanalysis}.
	
	\item \textbf{Online Dynamic Multiple Selection:} Suppose the data structure is first constructed on $n$ elements. (A close analysis of insert in \wref{sec:insertanalysis} shows that alternatively, we can construct the data structure on an empty set and achieve $O(1)$ time insertion before a query is performed.) After construction, a set of ranks $\{r_i\}$ are selected, specified online and in any order. Lazy search trees will support this selection in $O(B)$ time, where $B = \sum_{i=1}^m |\Delta_i| \log_2 (n/|\Delta_i|)$ is the lower bound for the multiple selection problem~\cite{KMMS}. We can further support additional insertions, deletions and queries. Data structures for online dynamic multiple selection were previously known~\cite{Barbay15,BARBAY16}, but the way we handle dynamism is more efficient, allowing for all the use cases mentioned here. We discuss this in \wref{sec:related}.
	
	\item \textbf{Split By Rank:} Lazy search trees can function as a data structure for repeated splitting by rank, supporting construction on an initial set of $n$ elements in $O(n)$ time, insertion into a piece of size $n$ in $O(\log \log n)$ time, and all splitting within a constant factor of the information-theoretic lower bound. Here, the idea is that we would like to support the operations insert and split at rank $r$, returning two pieces of a data structure of the same form. In a sense, this is a generalization of priority queues, where instead of extracting the minimum, we may extract the $k$ smallest elements, retaining the ability to perform further extractions on either of the two pieces returned. %
	As in scenario 3, the cost of splitting is $O(x \log c + \log n)$, where $x$ is the number of elements in the smaller resulting piece of the split, and we define $c$ so that the number of elements in the larger resulting piece of the split is $cx$. Again, $O(x \log c + \log n)$ is optimal. Note that we could also use an $O(\log \log n)$ time change-key operation for this application, though this time complexity only applies when elements are moved closer to the nearest split rank. If repeatedly extracting the $k$ smallest elements is desired, this corresponds to an $O(\log \log n)$ time decrease-key operation.
	
	\item \textbf{Incremental Quicksort:} A version of our data structure can perform splits internally via selecting random pivots with expected time complexity matching the bounds given in \wref{thm:main}. (We initially describe a version using exact selection, which is conceptually simpler but less practical.) The data structure can then be used to extract the $q$ smallest elements in sorted order, online in $q$, via an incremental quicksort. Here, $B = \Theta(q \log n)$ and our overall time complexity is $O(n + q \log n)$, which is optimal up to constant factors\footnote{Note that $n + q \log n = \Theta(n + q \log q)$. If the $q \log n$ term dominates, $q = \Omega(n/\log n)$ and so $\log n = \Theta(\log q)$.}. Previous algorithms for incremental sorting are known~\cite{Paredes06,Navarro10,Regla15,Aydin15}; however, our algorithm is extremely flexible, progressively sorting any part of the array in optimal time $O(B + n)$ while also supporting insertion, deletion, and efficient change-key. The heap operations insert and decrease-key are performed in $O(\log \log n)$ time instead of $O(\log n)$, compared to existing heaps based on incremental sorting~\cite{Navarro08,Navarro10}; see also~\cite{Edelkamp12,Brodal09}. Our data structure also uses only $O(\min(q, n))$ pointers, providing many of the same advantages of sorting-based heaps%
	. A more-complicated priority queue based on similar ideas to ours achieves Fibonacci heap amortized complexity with only a single extra word of space~\cite{Mortensen05}.
\end{enumerate}

We discuss the advantages and disadvantages of our model and data structure in the following subsections.

\subsection{Advantages}
\label{pro}

The advantages of lazy search trees are as follows:
\begin{enumerate}
	\item Superior runtimes to binary search trees can be achieved when queries are infrequent or non-uniformly distributed.
	\item A larger operation set, with the notable exception of efficient general merging, is made possible when used as a priority queue, supporting operations within an additive $O(n \log \log n)$ term of optimality, in our model.
	\item Lazy search trees can be implemented to use only $O(\min(q, n))$ pointers, operating mostly on arrays. This suggests smaller memory footprint, better constant factors, and better cache performance compared to many existing efficient priority queues or binary search trees. Our data structure is not built on the heap-ordered tree blueprint followed by many existing efficient priority queues~\cite{Fredman87,Fredman86,Chan09,Brodal12,Elmasry09,Haeupler11,Brodal96,Hansen15}. Instead, we develop a simple scheme based on unordered lists that may of independent interest. In particular, we are hopeful our data structure or adaptations thereof may provide a theoretically-efficient priority queue that gets around the practical inefficiencies associated with Fibonacci heaps~\cite{Fredman87} and its derivatives.

	\item While not a corollary of the model we consider, lazy search trees can be made to satisfy all performance theorems with regards to access time satisfied by splay trees. In this way, lazy search trees can be a powerful alternative to the splay tree. Locality of access can decrease both access and insertion times. This is discussed in \wref{sec:splay}.
\end{enumerate}

\subsection{Disadvantages}
\label{sec:con}

The weaknesses of lazy search trees are as follows:
\begin{enumerate}
	\item Our gap-based model requires inserted elements be placed in a gap immediately instead of delaying all insertion work until deemed truly necessary by query operations. In particular, a more powerful model would ensure that the number of comparisons performed on an inserted element depends only on the queries executed \emph{after} that element is inserted. There are operation sequences where this can make a $\Theta(\log n)$ factor difference in overall time complexity, but it is not clear whether this property is important on operation sequences arising in applications.
	
	\item We currently do not know whether the additive $O(\min( n \log q, n\log \log n))$ term in the complexity described in \wref{thm:main} over a sequence of insertions and queries is necessary. Fibonacci heaps and its variants show better performance is achievable in the priority queue setting.
	In \wref{sec:average}, we show the (essentially) $O(\log \log |\Delta_i|)$ terms for insertion and change-key can be improved to a small constant factor if the (new) rank of the element is drawn uniformly at random from valid ranks in $\Delta_i$. As a priority queue, this corresponds with operation sequences in which binary heaps~\cite{Williams64} provide constant time insertion.

	 \item The \emph{worst-case} complexity of a single \texttt{RankBasedQuery($r$)} can be $O(n)$. Further, unlike amortized search trees like the splay tree~\cite{Sleator85}, the average case complexity is not necessarily $O(\log n)$. By delaying sorting, our lower bound indicates that we may need to spend $\Theta(n)$ time to answer a query that splits a gap of size $|\Delta_i| = \Theta(n)$ into pieces of size $x$ and $cx$ for $c = \Theta(1)$. Further, aside from an initial $O(\log n)$ time search, the rest of the time spent during query is on writes, so that over the course of the operation sequence the number of writes is $\Theta(B + n)$. In this sense, our algorithm functions more similarly to a lazy quicksort than a red-black tree~\cite{Bayer72}, which requires only $\Theta(n)$ writes regardless of operation sequence.
\end{enumerate}

\subsection{Paper Organization}
We organize the remainder of the paper as follows. In the following section, \wref{sec:related}, we discuss related work. In \wref{sec:technical}, we give a high-level overview of the technical challenge. In \wref{sec:prelim}, we formalize the definition of the queries we support. In \wref{sec:bounds}, we discuss lower bounds in our gap-based model. In \wref{sec:ds}, we show how lazy search trees perform insertions, queries, deletions, and change-key operations. We analyze the costs of these operations in \wref{sec:analysis}. In \wref{sec:bulk}, we explain how binary search tree bulk-update operations split and merge can be performed on lazy search trees. We show in \wref{sec:average} that the complexity of insertion and change-key can be improved with a weak average-case assumption. In \wref{sec:random}, we show that exact selection in our query algorithm can be replaced with randomized pivoting while achieving the same expected time complexity. In \wref{sec:splay}, we show how splay trees can be used with lazy search trees and show that lazy search trees can be made to support efficient access theorems. We give concluding remarks, open problems, and briefly discuss a proof-of-concept implementation in \wref{sec:conclude}.

\section{Related Work}
\label{sec:related}

Lazy search trees unify several distinct research fields. The two largest, as previously discussed, are the design of efficient priority queues and balanced binary search trees. We achieved our result by developing an efficient priority queue and lazy binary search tree simultaneously. There are no directly comparable results to our work, but research in \emph{deferred data structures} and \emph{online dynamic multiple selection} comes closest. We further discuss differences between dynamic optimality and our work.

\subsection{Deferred Data Structures}
\label{sec:deferred-data-structures}

To our knowledge, the idea of deferred data structures was first proposed by Karp, Motwani, and Raghavan in 1988~\cite{Karp88}. Similar ideas have existed in slightly different forms for different problems~\cite{Smid89,Borodin81,Brodal11,Barbay2019,Barbay17,Ar02,Gum01,Aggarwal91}. 
The term ``deferred data structure'' has been used more generally for delaying processing of data until queries make it necessary, but we focus on works for one-dimensional data here, as it directly pertains to the problem we consider.

Karp, Motwani and Raghavan~\cite{Karp88} study the problem of answering membership queries on a static, unordered set of $n$ elements in the comparison model. One solution is to construct a binary search tree of the data in $O(n \log n)$ time and then answer each query in $O(\log n)$ time. This is not optimal if the number of queries is small. Alternatively, we could answer each query in $O(n)$ time, but this is clearly not optimal if the number of queries is large. Karp et al. determine the lower bound of $\Omega((n+q) \log (\min(n,q))) = \Omega(n \log q + q \log n)$ time to answer $q$ queries on a static set of $n$ elements in the worst case and develop a data structure that achieves this complexity.

This work was extended in 1990 to a dynamic model. Ching, Melhorn, and Smid show that $q'$ membership queries, insertions, and deletions on an initial set of $n_0$ unordered elements can be answered in $O(q' \log (n_0+q') + (n_0+q') \log q') = O(q' \log n_0 + n_0 \log q')$ time~\cite{CHING90}. When membership, insertion, and deletion are considered as the same type of operation, this bound is optimal. 

It is not very difficult 
(although not explicitly done in~\cite{CHING90}) to modify the result of Ching et al.\ to obtain a data structure supporting $n$ insertions and $q''$ membership or deletion operations in $O(q'' \log n + n \log q'')$ time, the runtime we achieve for uniform queries. We will see in \wref{sec:technical} that the technical difficulty of our result is to achieve the fine-grained complexity based on the query-rank distribution. For more work in one-dimensional deferred data structures, see~\cite{Smid89,Borodin81,Brodal11,Barbay2019,Barbay17,Gum01}.

\subsection{Online Dynamic Multiple Selection}

The optimality of Karp et al.~\cite{Karp88} and Ching et al.~\cite{CHING90} is in a model where the ranks requested of each query are not taken into account. In the multiple selection problem, solutions have been developed that consider this information in the analysis. Suppose we wish to select the elements of ranks $r_1 < r_2 < \cdots < r_q$ amongst a set of $n$ unordered elements. Define $r_0 = 0$, $r_{q+1} = n$, and $\Delta_i$ as the set of elements of rank greater than $r_{i-1}$ and at most $r_i$. Then $|\Delta_i| = r_i - r_{i-1}$ and as in \wref{thm:main}, $B = \sum_{i=1}^m |\Delta_i| \log_2(n/|\Delta_i|)$. The information-theoretic lower bound for multiple selection is $B-O(n)$ comparisons~\cite{Dobkin81}. Solutions have been developed that achieve $O(B + n)$ time complexity~\cite{Dobkin81} or $B + o(B) + O(n)$ comparison complexity~\cite{KMMS}.

The differences between the multiple selection problem and deferred data structuring for one-dimensional data are minor. Typically, deferred data structures are designed for online queries, whereas initial work in multiple selection considered the setting when all query ranks are given at the same time as the unsorted data. Solutions to the multiple selection problem where the ranks $r_1, \ldots, r_q$ are given online and in any order have also been studied, however~\cite{Barbay13}. Barbay et al.~\cite{Barbay15,BARBAY16} further extend this model to a dynamic setting: They consider online dynamic multiple selection where every insertion is preceded by a search for the inserted element. Deletions are ultimately performed in $O(\log n)$ time. Their data structure uses $B + o(B) + O(n + q'\log n)$ comparisons, where $q'$ is the number of search, insert, and delete operations. %
The crucial difference between our solution and that of Barbay et al.~\cite{Barbay15,BARBAY16} is how we handle insertions. Their analysis assumes every insertion is preceded by a search and therefore insertion must take $\Omega(\log n)$ time. Thus, for their result to be meaningful (\ie, allow $o(n \log n)$ performance), the algorithm must start with an initial set of $n_0 = n\pm o(n)$ elements. While Barbay et al.\ focus on online dynamic multiple selection algorithms with near-optimal comparison complexity, the focus of lazy search trees is on generality. We achieve similar complexity as a data structure for online multiple selection while also achieving near-optimal performance as a priority queue. We discuss the technical challenges in achieving this generality in \wref{sec:technical}.

\subsection{Dynamic Optimality}

As mentioned, the dynamic optimality conjecture has received vast attention in the past four decades~\cite{Allen78,Demaine09,Demaine07,Iacono16,Bose2020,Sleator83,Wilber89,Kozma2019,Chalermsook15}. The original statement conjectures that the performance of the splay tree is within a constant factor of the performance of any binary search tree on any sufficiently long access sequence~\cite{Sleator85}. To formalize this statement, in particular the notion of ``any binary search tree'', the BST model of computation has been introduced, forcing the data structure to take the form of a binary tree with access from the root and tree rotations for updates. Dynamic optimality is enticing because it conjectures splay trees~\cite{Sleator85} and a related ``greedy BST''~\cite{Demaine09} to be within a constant factor of optimality on \emph{any} sufficiently long access sequence. This \emph{per-instance} optimality~\cite{Fagin2003} is more powerful than the sense of optimality used in less restricted models, where it is often unattainable. Any sorting algorithm, for example, must take $\Omega(n \log n)$ time in the \emph{worst case}, but on any particular input permutation, an algorithm designed to first check for that specific permutation can sort it in $O(n)$ time: simply apply the inverse permutation and check if the resulting order is monotonic. %

The bounds we give in \wref{sec:bounds} are w.\,r.\,t.\ the \emph{worst case} over operation sequences based on distribution of gaps $\{\Delta_i\}$,
but hold for \emph{any} comparison-based data structure. 
Hence, lazy search trees achieve a weaker notion of optimality compared to dynamic optimality, 
but do so against a vastly larger class of algorithms.

Since splay trees, greedy BSTs, and lazy search trees are all implementations of sorted dictionaries and conjectured dynamically optimal, it is insightful to contrast the access theorems of dynamically-optimal BSTs with the improvements given in \wref{thm:main}. 
Superficially, the two notions are orthogonal, with dynamic optimality allowing only queries, and
our bound becoming interesting mostly when insertions and queries are mixed.
On the other hand, the form of performance improvements achievable are indeed quite similar, 
as the following property shows.

\begin{definition}[Static Optimality~\cite{Knuth73,Allen78,Sleator85}]
\label{def:statico}
	Let $S$ denote the set of elements in the data structure and let $q_x$ denote the number of times element $x$ is accessed in a sequence of $m$ accesses. Assume every element is accessed at least once. 
	A data structure is said to achieve static optimality if the cost to perform any such access sequence is
	\[
	O(m + \sum_{x \in S} q_x \log(m/q_x)).
	\]
\end{definition}
Historically, research into optimal binary search trees started with this notion of static optimality,
and both splay trees and greedy BSTs have been shown to be statically optimal~\cite{Sleator85,Fox2011}.
Contrast the bound given in \wref{def:statico} with the bound $O(B + n)$, where again we define $B = \sum_{i=1}^m |\Delta_i| \log_2(n/|\Delta_i|)$. If we replace $q_x$ and $m$ in \wref{def:statico} with $|\Delta_i|$ and $n$, respectively, they are exactly the same: the savings for query costs arising from repeated accesses with nonuniform access probabilities equal the savings for insertion costs when query ranks are nonuniform.

\section{Technical Overview}
\label{sec:technical}

This research started with the goal of generalizing a data structure that supports $n$ insertions and $q \leq n$ rank-based queries in $O(n \log q)$ time. Via a reduction from multiple selection, $\Omega(n \log q)$ comparisons are necessary in the worst case. However, by applying the fine-grained analysis based on rank distribution previously employed in the multiple selection literature~\cite{Dobkin81}, a new theory which generalizes efficient priority queues and binary search trees is made possible.

As will be discussed in \wref{sec:bounds}, to achieve optimality on sequences of insertion and distinct queries with regards to the fine-grained multiple selection lower bound, insertion into gap $\Delta_i$ should take $O(\log(n/|\Delta_i|))$ time. A query which splits a gap $\Delta_i$ into two gaps of sizes $x$ and $cx$ ($c \geq 1$), respectively, should take $O(x \log c + \log n)$ time. These complexities are the main goals for the design of the data structure.

The high-level idea will be to maintain elements in a gap $\Delta_i$ in an auxiliary data structure (the \emph{interval data structure} of \wref{sec:ds}). All such auxiliary data structures are then stored in a biased search tree so that access to the $i$th gap $\Delta_i$ is supported in $O(\log(n/|\Delta_i|))$ time. This matches desired insertion complexity and is within the $O(\log n)$ term of query complexity. The main technical difficulty is to support efficient insertion and repeated splitting of the auxiliary data structure.

Our high-level organization is similar to the selectable sloppy heap of Dumitrescu~\cite{Dumitrescu19}. The difference is that while the selectable sloppy heap keeps fixed quantile groups in a balanced search tree and utilizes essentially a linked-list as the auxiliary data structure, in our case the sets of elements stored are dependent on previous query ranks, the search tree is biased, and we require a more sophisticated auxiliary data structure.

Indeed, in the priority queue case, the biased search tree has a single element $\Delta_1$, and all operations take place within the auxiliary data structure. Thus, we ideally would like to support $O(1)$ insertion and $O(x\log c)$ split into parts of size $x$ and $cx$ ($c \geq 1$) in the auxiliary data structure. If the number of elements in the auxiliary data structure is $|\Delta_i|$, we can imagine finding the minimum or maximum as a split with $x = 1$ and $c = |\Delta_i|-1$, taking $O(\log |\Delta_i|)$ time. However, the ability to split at any rank in optimal time complexity is not an operation typically considered for priority queues. Most efficient priority queues store elements in heap-ordered trees, providing efficient access to the minimum element but otherwise imposing intentionally little structure so that insertion, decrease-key, and merging can all be performed efficiently.

Our solution is to group elements within the auxiliary data structure in the following way. We separate elements into groups (``intervals'') of unsorted elements, but the elements between each group satisfy a total order. Our groups are of exponentially increasing size as distance to the gap boundary increases. Within a gap $\Delta_i$, we maintain $O(\log |\Delta_i|)$ such groups. Binary search then allows insertion and key change in $O(\log \log |\Delta_i|)$ time. While not $O(1)$, the structure created by separating elements in this way allows us to split the data structure in about $O(x)$ time, where $x$ is the distance from the split point to the closest query rank. Unfortunately, several complications remain.

Consider if we enforce the exponentially-increasing group sizes in the natural way in data structure design. That is, we select constants $c_1\le c_2$ such that as we get farther from the gap boundary, the next group is at least a factor $c_1 > 1$ larger than the previous but at most a factor $c_2$. We can maintain this invariant while supporting insertion and deletion, but splitting is problematic. After splitting, we must continue to use both pieces as a data structure of the same form. However, in the larger piece, the $x$ elements removed require restructuring not only the new closest group to the gap boundary but could require a cascading change on all groups. Since the elements of each group are unstructured, this cascading change could take $\Omega(|\Delta_i|)$ time.

Thus, we must use a more flexible notion of ``exponentially increasing" that does not require significant restructuring after a split. This is complicated by guaranteeing fast insertion and fast splits in the future. In particular, after a split, if the larger piece is again split close to where the previous split occurred, we must support this operation quickly, despite avoiding the previous cascading change that would guarantee this performance. Further, to provide fast insertion, we must keep the number of groups at $O(\log |\Delta_i|)$, but after a split, the best way to guarantee fast future splits is to create more groups.

We will show that it is possible to resolve all these issues and support desired operations efficiently by applying amortized analysis with a careful choice of structure invariants. While we do not achieve $O(1)$ insertion and decrease-key cost, our data structure is competitive as an efficient priority queue while having to solve the more complicated issues around efficient repeated arbitrary splitting.%

\section{Rank-Based Queries}
\label{sec:prelim}

We formalize operation \texttt{RankBasedQuery($r$)} as follows. We first describe what we call an \emph{aggregate function}.

\begin{definition}[Aggregate function]
	Let $S$ be a multiset of comparable elements and let $f(S)$ be a function\footnote{We do not actually require a strict function $f(S) = y$ for a set $S$, but rather can let the aggregate function depend on the queries that dynamically change that set. In particular, we can (initially) map $f(S) = \min S$ or $f(S) = \max S$ and change this value to decrease/increase monotonically as $S$ is updated, even if the minimum/maximum is removed.} computed on those elements. Suppose $S'$ is such that $S'$ differs from $S$ by the addition or removal of element $x$. Let $n = \max(|S|, |S'|)$. Then $f$ is an \textit{aggregate function} maintainable in $g(n)$ time if $f(S')$ can be computed from $f(S)$ and $x$ in $g(n)$ time.
\end{definition}

We focus on aggregates with $g(n) = O(1)$, though in principle any $g(n)$ can be supported with appropriate changes to overall runtime complexity. We formalize \emph{rank-based queries} as follows.

\begin{definition}[Rank-based query]
	\label{def:rankbasedquery}
	Call a query on a multiset of comparable elements $S$ such that $|S| = n$ a \textit{rank-based query} pertaining to rank $r$ if the following two conditions are satisfied:
	\begin{enumerate}
		\item Consider if $S$ is split into two sets $X$ and $Y$ such that for all $x \in X$, $y \in Y$, $x \leq y$. It must be possible, based on an aggregate function $f$ on $X$ and $Y$, to reduce the query to a query that can be answered considering only the elements of $X$ or only the elements of $Y$. The query rank $r$ should be such that if the query is reduced to $X$, $r \leq |X|$, and if the query is reduced to $Y$, $|X| < r$.
		\item It must be possible to answer the query on $S$ in $O(n)$ time.
	\end{enumerate}
\end{definition}

Critical to our analysis is the rank $r$ associated with each \texttt{RankBasedQuery($r$)} operation. We associate with each operation a rank $r$ which must be contained in each subproblem according to a recursion based on \wref{def:rankbasedquery}. Amongst a set of unsorted elements, $r$ can be chosen arbitrarily, but whichever rank is chosen will affect the restructuring and change the complexity of future operations. Implementation decisions may change $r$ to be $r-1$ or $r+1$; such one-off errors do not have measured effect on complexity, as long as the extract minimum or extract maximum queries result in a single gap $\Delta_1$.

The following well-studied operations fit our definition of rank-based query with ranks $r$ as described;
the aggregate function is either the cardinality of the set or a range of keys for the set.
\begin{itemize}
	\item \texttt{Rank($k$)} $\ce$ Determine the rank of key $k$. Rank $r$ is the rank of $k$ in $S$.
	\item \texttt{Select($r$)} $\ce$ Select the element of rank $r$ in $S$. Rank $r$ is the rank selected.
	\item \twopart{\texttt{Contains($k$)} $\ce$ }{Determine if an element $(k,v)$ is represented in $S$ (and if so, returns $v$). Rank $r$ is the rank of $k$ in $S$.}
	\item \twopart{\texttt{Successor($k$)} $\ce$ }{Determine the successor of key $k$ in $S$. Rank $r$ is the rank of the successor of~$k$.}
	\item \twopart{\texttt{Predecessor($k$)} $\ce$ }{Determine the predecessor of key $k$ in $S$. Rank $r$ is the rank of the predecessor of $k$.}
	\item \texttt{Minimum()} $\ce$ Return a minimum element of $S$. Rank $r$ is $1$.
	\item \texttt{Maximum()} $\ce$ Return a maximum element of $S$. Rank $r$ is $n$.
\end{itemize}
On edge cases where the successor or predecessor does not exist, we can define $r$ to be $n$ or $1$, respectively. Similarly, in the case $(k,v)$ is represented in $S$ on a \texttt{Rank($k$)} or \texttt{Contains($k$)} query, we must pick a tie-breaking rule for rank $r$ returned consistent with the implemented recursion following \wref{def:rankbasedquery}.

\section{Lower and Upper Bounds}
\label{sec:bounds}

The balanced binary search tree is the most well-known solution to the sorted dictionary problem. It achieves $O(\log n)$ time for a rank-based query and $O(\log n)$ time for all dynamic operations. Via a reduction from sorting, for a sequence of $n$ arbitrary operations, $\Omega(n \log n)$ comparisons and thus $\Omega(n \log n)$ time is necessary in the worst case.

However, this time complexity can be improved by strengthening our model. The performance theorems of the splay tree~\cite{Sleator85} show that although $\Omega(q \log n)$ time is necessary on a sequence of $q$ arbitrary queries on $n$ elements, many access sequences can be answered in $o(q \log n)$ time. Our model treats sequences of element \emph{insertions} similarly to the splay tree's treatment of sequences of element access. Although $\Omega(n \log n)$ time is necessary on a sequence of $n$ insert or query operations, on many operation sequences, $o(n \log n)$ time complexity is possible, as the theory of efficient priority queues demonstrates.

Our complexities are based on the distribution of elements into the set of gaps $\{\Delta_i\}$. We can derive a lower bound on a sequence of operations resulting in a set of gaps $\{\Delta_i\}$ via reducing multiple selection to the sorted dictionary problem. We prove \wref{thm:lb} below.

\begin{proof}[Proof of \wref{thm:lb}{}]
We reduce multiple selection to the sorted dictionary problem. The input of multiple selection is a set of $n$ elements and ranks $r_1 < r_2 < \cdots < r_q$. We are required to report the elements of the desired ranks. We reduce this to the sorted dictionary problem by inserting all $n$ elements in any order and then querying for the desired ranks $r_1, \ldots, r_q$, again in any order.

Define $r_0 = 0$, $r_{q+1} = n$, and $\Delta_i$ as the set of elements of rank greater than $r_{i-1}$ and at most $r_i$. (This definition coincides with the gaps resulting in our data structure when query rank $r$ falls in the new gap $\Delta'_i$, described in \wref{sec:model}.) Then $|\Delta_i| = r_i - r_{i-1}$ and as in \wref{thm:main}, $B = \sum_{i=1}^m |\Delta_i| \log_2(n/|\Delta_i|)$. Note that here, $m = q+1$. The information-theoretic lower bound for multiple selection is $B-O(n)$ comparisons~\cite{Dobkin81}. Since any data structure must spend $\Omega(n)$ time to read the input, this also gives a lower bound of $\Omega(B + n)$ time. This implies the sorted dictionary problem resulting in a set of gaps $\{\Delta_i\}$ must use at least $B-O(n)$ comparisons and take $\Omega(B + n)$ time, in the worst case.
\end{proof}

\begin{remark}[Multiple selection inputs]
For the operation sequence from the proof of \wref{thm:lb},
\wref{thm:main} states our performance as $O(B + \min(n \log q, n\log\log n))$. A closer examination of our data structure in \wref{sec:insertanalysis} shows we actually achieve $O(B+n)$ complexity on such sequences, since insertions performed before any queries actually take $O(1)$ time.
\end{remark}

To achieve the performance stated in \wref{thm:lb} on any operation sequence, we will first consider how the bound $\Omega(B + n)$ changes with insertions and queries. This will dictate the allotted (amortized) time we can spend per operation to achieve an optimal complexity over the entire operation sequence.

We give the following regarding insertion time; recall our convention from \wpref{fn:note1} that
$\log(x) = \max(\log_2(x),1)$ and $\log_2$ is the binary logarithm.

\begin{lemma}[Influence of insert on lower bound]
	\label{lem:inserttime}
	Suppose we insert an element into gap $\Delta_i$. Then the bound $\Omega(B + n)$ increases by $\Omega(\log(n/|\Delta_i|))$.
\end{lemma}

\begin{proof}
	The insertion simultaneously increases $|\Delta_i|$ and $n$, but we will consider the effect of these changes
	separately. We first keep $n$ unchanged and 
	consider how $B$ changes in gap $\Delta_i$. Before insertion, the contribution to $B$ for gap $\Delta_i$ is $|\Delta_i| \log_2 (n/|\Delta_i|)$; after the insertion it is $(|\Delta_i|+1) \log_2 (n/(|\Delta_i|+1))$. Therefore, the change is
	\begin{equation}
	\label{dB}
	(|\Delta_i|+1) \log_2 (n/(|\Delta_i|+1)) - |\Delta_i| \log_2 (n/|\Delta_i|).
	\end{equation}
	Consider the function $f(x) = x\log_2(n/x)$, where we treat $n$ as a constant. Then \eqref{dB} is at least the minimum value of the derivative $f'(x)$ with $x \in [|\Delta_i|, |\Delta_i|+1]$. The derivative of $f(x)$ is $f'(x) = -\log_2(e) + \log_2(n/x)$. This gives that the change in $B$ is at least $-\log_2(e) + \log_2(n/|\Delta_i|)$.
	
	Now consider the effect of making $n$ one larger. This will only increase $B$; by the bound $\Omega(B + n)$, this change is (at least) $\Omega(1)$. We may therefore arrive at an increase of $\Omega(\log_2(n/|\Delta_i|)+1) = \Omega(\log(n/|\Delta_i|))$.
\end{proof}

\wref{lem:inserttime} implies that optimal insertion complexity is $\Omega(\log(n/|\Delta_i|))$.
This bound is using the fact the change in the set of gaps $\{\Delta_i\}$ resulting from an insertion corresponds to a multiple selection problem with lower bound greater by $\Omega(\log(n/|\Delta_i|))$. Since the multiple selection problem itself has insertions preceding queries, this lower bound is in some sense artificial. However, we can alternatively consider the problem of determining in which gap an inserted element falls. Here, information theory dictates complexities of $\Omega(\log m)$ if each gap is weighted equally or $\Omega(\log(n/|\Delta_i|))$ if gap $\Delta_i$ is weighted with weight $|\Delta_i|$~\cite{Bent85}. The latter corresponds with the change in $B$ noted above.

We now give the following regarding query time.

\begin{lemma}[Influence of query on lower bound]
	\label{lem:querytime}
	Suppose a query splits a gap $\Delta_i$ into two gaps of size $x$ and $cx$, respectively, with $c \geq 1$. Then the bound $\Omega(B + n)$ increases by $\Omega(x\log c)$.
\end{lemma}
\begin{proof}
	The change in $B$ is 
	\begin{equation}
	\label{dB2}
	x \log_2\left(\frac{n}{x}\right) + cx \log_2\left(\frac{n}{cx}\right) - (c+1)x \log_2\left(\frac{n}{(c+1)x}\right).
	\end{equation}
	By manipulating logarithms and canceling terms, we can rearrange \eqref{dB2} to $x((c+1) \log_2 (c+1) - c \log_2 c)$, which is greater than $x \log_2 (c+1)$. Thus the increase in $\Omega(B + n)$ is $\Omega(x \log c)$.
\end{proof}

\wref{lem:querytime} gives a lower bound of $\Omega(x \log c)$ per rank-based query operation. Here, the bound is not artificial in any sense: insertions precede queries in the reduction of multiple selection to the sorted dictionary problem. We must spend time $\Omega(x \log c)$ to answer the query as more queries may follow and the total complexity must be $\Omega(B + n)$ in the worst case.

We can improve the query lower bound by considering the effect on $B$ over a sequence of gap-splitting operations. Consider the overall bound $B = \sum_{i=1}^m |\Delta_i| \log_2(n/|\Delta_i|)$. It can be seen that $B = \Omega(m \log n)$. Therefore, we can afford amortized $O(\log n)$ time whenever a new gap is created, even if it is a split say with $x = 1$, $c = 1$.%

Consider the lower bound given by the set of gaps $\{\Delta_i\}$ in \wref{thm:lb} combined with the above insight that queries must take $\Omega(\log n)$ time. If query distribution is not considered, the worst case is that $|\Delta_i| = \Theta(n/q)$ for all $i$. Then $B + q \log n = \Omega(n \log q + q \log n)$. This coincides with the lower bound given in~\cite{Karp88}.

It is worth noting that both \wref{lem:inserttime} and \wref{lem:querytime} can also be proven by information-theoretic arguments, without appealing to the algebraic bound $B$ given in multiple selection.  The number of comparisons to identify the $x$ largest elements in a set of $(c+1)x$ elements is $\log{ \binom{(c+1)x}x}$, which is $\Omega(x \log c)$. A similar argument can be made that increasing $n$ by $1$ and category $\Delta_i$ by $1$ implies the number of comparisons required of the underlying selection problem increases by $\Omega(\log(n/\Delta_i))$.

\section{Data Structure}
\label{sec:ds}

We are now ready to discuss the details of lazy search trees. The high-level idea was discussed in \wref{sec:technical}. The data structure as developed is relatively simple, though it requires a somewhat tricky amortized time analysis given in the following section.

We split the data structure into two levels. At the top level, we build a data structure on the set of gaps $\{\Delta_i\}$. In the second level, actual elements are organized into a set of \emph{intervals} within a gap. Given a gap $\Delta_i$, intervals within $\Delta_i$ are labeled $\mathcal{I}_{i,1}, \mathcal{I}_{i,2}, \ldots, \mathcal{I}_{i,\ell_i}$, with $\ell_i$ the number of intervals in gap $\Delta_i$. The organization of elements of a gap into intervals is similar to the organization of elements into a gap. Intervals partition a gap by rank, so that for elements $x \in \mathcal{I}_{i,j}$, $y \in \mathcal{I}_{i,j+1}$, $x \leq y$. Elements within an interval are unordered. By convention, we will consider both gaps and intervals to be ordered from left to right in increasing rank. A graphical sketch of the high-level decomposition is given in \wref{fig:decomposition}.

\begin{figure}[htbp]
	\def\intsep{0.08}
	\newcommand\drawint[4][1]{
		\draw[semithick,|-|] (#2,0) -- 
			node[below=2pt,font=\scriptsize,scale=.75] {\scalebox{#1}[1]{$\mathcal{I}_{#4}$}} 
			++(#3-\intsep,0) ;
	}
	\begin{tikzpicture}[scale=.32]
		\begin{scope}
			\drawint[.7]{0} 1{1,1}
			\drawint    {1} 2{1,2}
			\drawint    {3} 4{1,3}
			\drawint    {7} 4{1,4}
			\drawint    {11}2{1,5}
			\drawint[.7]{13}1{1,6}
		\end{scope}
		\begin{scope}[shift={(14+.5,0)}]
			\drawint[.7]{0} 1{2,1}
			\drawint    {1} 2{2,2}
			\drawint    {3} 3{2,3}
			\drawint    {6} 6{2,4}
			\drawint    {12}5{2,5}
			\drawint    {17}3{2,6}
			\drawint    {20}2{2,7}
			\drawint[.7]{22}1{2,8}
		\end{scope}
		\begin{scope}[shift={(14+23+1,0)}]
			\drawint[.7]{0} 1{3,1}
			\drawint    {1} 5{3,2}
			\drawint    {6} 3{3,3}
			\drawint[.7]{9} 1{3,4}
		\end{scope}
	    \foreach \x/\l in {0,14.5,14+23+1,14+23+10+1.5} {
			\draw[thick,gray] (\x-.25,-1) -- (\x-.25,3) ;
		}
		\foreach \f/\t/\l in {0/14/1,14.5/37.5/2,38/48/3} {
			\draw[decoration={brace,amplitude=5pt},decorate]
				(\f,1.5) -- node[above=3pt,font=\scriptsize] {$\Delta_{\l}$} ++(\t-\f,0) ;
		}
		\node[font=\scriptsize,anchor=east] at (-0.5,2.3) {Gaps:} ;
		\node[overlay,font=\scriptsize,anchor=east] at (-0.5,0) {Intervals:} ;
	\end{tikzpicture}
\caption{The two-level decomposition into gaps $\{\Delta_i\}$ and intervals $\{\mathcal{I}_{i,j}\}$.}
  \label{fig:decomposition}
\end{figure}
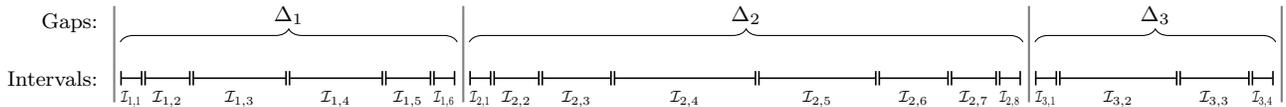

\subsection{The Gap Data Structure}

We will use the following data structure for the top level.

\begin{lemma}[Gap Data Structure]
\label{lem:gapds}
	There is a data structure for the set of gaps $\{\Delta_i\}$ that supports the following operations in the given worst-case time complexities. Note that $\sum_{i=1}^m |\Delta_i| = n$.
	\begin{enumerate}
	  \item Given an element $e=(k,v)$, determine the index $i$ such that $k \in \Delta_i$, in $O(\log(n/|\Delta_i|))$ time.
	  \item Given a $\Delta_i$, increase or decrease the size of $\Delta_i$ by $1$, adjusting $n$ accordingly, in $O(\log(n/|\Delta_i|))$ time.
	  \item Remove $\Delta_i$ from the set, in $O(\log n)$ time.
	  \item Add a new $\Delta_i$ to the set, in $O(\log n)$ time.
	\end{enumerate}
	It is also possible to store aggregate functions within the data structure (on subtrees), as required by some queries that fit \wref{def:rankbasedquery}.
\end{lemma}
\begin{proof}
We can use, for example, a globally-biased $2, b$ tree~\cite{Bent85}. We assign gap $\Delta_i$ the weight $w_i = |\Delta_i|$; the sum of weights, $W$, is thus equal to $n$. Access to gap $\Delta_i$, operation~1, is handled in $O(\log(n/|\Delta_i|))$ worst-case time~\cite[Thm.\,1]{Bent85}. By~\cite[Thm.\,11]{Bent85}, operation~2 is handled via weight change in $O(\log(n/|\Delta_i|))$ worst-case time. Again by~\cite[Thm.\,11]{Bent85}, operations~3 and~4 are handled in $O(\log n)$ worst-case time or better.
\end{proof}

\begin{remark}[Alternative implementations]
A variety of biased search trees can be used as the data structure of \wref{lem:gapds}. In \wref{sec:splay}, we suggest splay trees for that purpose, which greatly simplifies implementation at the cost of making the runtimes amortized. What is more, we show that efficient access properties of the data structure of \wref{lem:gapds} can be inherited by the lazy search tree, hence the (orthogonal) efficiency gains for insertions in lazy search trees and for structured access sequences in splay trees can be had simultaneously.
%
%
\end{remark}

The top level data structure allows us to access a gap in the desired time complexity for insertion. However, we must also support efficient queries. In particular, we need to be able to split a gap $\Delta_i$ into two gaps of size $x$ and $cx$ ($c \geq 1$) in amortized time $O(x \log c)$. We must build additional structure amongst the elements in a gap to support such an operation efficiently. At the cost of this organization, in the worst case we pay an additional $O(\log \log |\Delta_i|)$ time on insertion and key-changing operations.

\subsection{The Interval Data Structure}
\label{sec:intervalds}

We now discuss the data structure for the intervals. 
Given a gap $\Delta_i$, intervals $\mathcal{I}_{i,1}, \mathcal{I}_{i,2}, \ldots, \mathcal{I}_{i,\ell_i}$ are contained within it and maintained in a data structure as follows. 
We maintain with each interval the two splitting keys $(k_l,k_r)$ that separate this interval from its predecessor and successor (using $-\infty$ and $+\infty$ for the outermost ones), respectively;
the interval only contains elements $e=(k,v)$ with $k_l \le k \le k_r$.
We store intervals in sorted order in an array (see \wref{rem:avoid-arrays}), sorted with respect to $(k_l,k_r)$.
We can then find an interval containing a given key~$k$, \ie, with $k_l \le k \le k_r$,
using binary search in $O(\log \ell_i)$ time.

\begin{remark}[Handling duplicate keys]
	\label{rem:duplicate-keys-search}
	Recall that we allow repeated insertions, \ie, elements with the same key $k$.
	As detailed in \wref{sec:queryalg}, intervals separated by a splitting key $k$ can then both contain
	elements with key $k$. To guide the binary search in these cases,
	we maintain for each interval the number of elements with keys equal to the splitting keys $k_l$ and $k_r$.
\end{remark}

As we will see below, the number of intervals in one gap is always $O(\log n)$, 
and only changes during a query, so we can afford to update this array on query in linear time.

\begin{remark}[Avoiding arrays]
\label{rem:avoid-arrays}
Note that, to stay within the pointer-machine model, we can choose to arrange the intervals within any gap in a balanced binary search tree, thus providing the binary search without array accesses.
This also allows the ability to add new intervals efficiently.
In practice, however, binary search on an array is likely to be preferred.
\end{remark}

We conceptually split the intervals into two groups: intervals on the \textit{left} side and intervals on the \textit{right} side. An interval is defined to be in one of the two groups by the following convention.
\begin{enumerate}[label=(\Alph*),font=\bfseries]
\item \label{rule:left-right}
	\textbf{Left and right intervals:} An interval $\mathcal{I}_{i,j}$ in gap $\Delta_i$ is on the \emph{left side} if the closest query rank (edge of gap $\Delta_i$ if queries have occurred on both sides of $\Delta_i$) is to the left. Symmetrically, an interval $\mathcal{I}_{i,j}$ is on the \emph{right side} if the closest query rank is on the right. An interval with an equal number of elements in $\Delta_i$ on its left and right sides can be defined to be on the left or right side arbitrarily.
\end{enumerate}

Recall the definition of closest query rank stated in \wref{fn:note2}. The closest query rank is the closest boundary of gap $\Delta_i$ that was created in response to a query.

We balance the sizes of the intervals within a gap according to the following rule:
\begin{enumerate}[label=(\Alph*),start=2,font=\bfseries]
\item \label{rule:merge}
	\textbf{Merging intervals:}
	Let $\mathcal{I}_{i,j}$ be an interval on the left side, not rightmost of left side intervals. We merge $\mathcal{I}_{i,j}$ into adjacent interval to the right, $\mathcal{I}_{i,j+1}$, if the number of elements left of $\mathcal{I}_{i,j}$ in $\Delta_i$ equals or exceeds $|\mathcal{I}_{i,j}| + |\mathcal{I}_{i,j+1}|$. We do the same, reflected, for intervals on the right side.
\end{enumerate}
The above rule was carefully chosen to satisfy several components of our analysis. As mentioned, we must be able to answer a query for a rank $r$ near the edges of $\Delta_i$ efficiently. This implies we need small intervals near the edges of gap $\Delta_i$, since the elements of each interval are unordered. However, we must also ensure the number of intervals within a gap does not become too large, since we must determine into which interval an inserted element falls at a time cost outside of the increase in $B$ as dictated in \wref{lem:inserttime}. We end up using the structure dictated by \wref{rule:merge} directly in our analysis of query complexity, particularly in \wref{sec:ensure}.

Note that \wref{rule:merge} causes the loss of information. Before a merge, intervals $\mathcal{I}_{i,j}$ and $\mathcal{I}_{i,j+1}$ are such that for any $x \in \mathcal{I}_{i,j}$ and $y \in \mathcal{I}_{i,j+1}$, $x \leq y$. After the merge, this information is lost. Surprisingly, this does not seem to impact our analysis. Once we pay the initial $O(\log \ell_i)$ cost to insert an element via binary search, the merging of intervals happens seldom enough that no additional cost need be incurred.

\wref{rule:merge} ensures the following.

\begin{lemma}[Few intervals]
\label{lem:numintervals}
Within a gap $\Delta_i$, there are at most $4 \log(|\Delta_i|)$ intervals.
\end{lemma}
\begin{proof}
First consider intervals on the left side. Let intervals $\mathcal{I}_{i, j+1}$ and $\mathcal{I}_{i, j+2}$ be on the left side. It must be that the number of elements in intervals $\mathcal{I}_{i, j+1}$ and $\mathcal{I}_{i, j+2}$ together is equal to or greater than the number of elements in the first $j$ intervals, by \wref{rule:merge}. 
Indeed, the worst-case sequence of interval sizes is $1,1,1,2,2,4,4,8,8,16,16,\ldots$, obtained recursively as $a_1=a_2=1$ and $a_j = a_1+\cdots+a_{j-2}+1-a_{j-1}$.
It follows that with every two intervals, the total number of elements at least doubles;
indeed we can show that the first $k$ intervals contain at least $(\sqrt 2)^{k+2}$ elements,
therefore $n$ elements are spread over at most $2 \log_2 n -2$ intervals.
To count intervals on the left resp.\ right side in $\Delta_i$, we observe that the maximal
number of intervals occurs if half of the elements are on either side, so there
can be at most $2 \cdot (2\log_2 (|\Delta_i|/2)-2) \le 4\log(|\Delta_i|)$ intervals in gap $\Delta_i$.
\end{proof}

For ease of implementation, we will invoke \wref{rule:merge} only when a \emph{query} occurs in gap $\Delta_i$. In the following subsection, we will see that insertion does not increase the number of intervals in a gap, therefore \wref{lem:numintervals} will still hold at all times even though \wref{rule:merge} might temporarily be violated after insertions. We can invoke \wref{rule:merge} in $O(\log |\Delta_i|)$ time during a query, since $|\Delta_i| \leq n$ and we can afford $O(\log n)$ time per query.

\subsection{Representation of Intervals}

It remains to describe how a single interval is represented internally.
Our analysis will require that merging two intervals can be done in $O(1)$ time and further that deletion from an interval can be performed in $O(1)$ time ($O(\log n)$ time actually suffices for $O(\log n)$ time delete overall, but on many operation sequences the faster interval deletion will yield better runtimes). Therefore, the container in which elements reside in intervals should support such behavior. An ordinary linked list certainly suffices; however, we can limit the number of pointers used in our data structure by representing intervals as a linked list of arrays. Whenever an interval is constructed, it can be constructed as a single (expandable) array. As intervals merge, we perform the operation in $O(1)$ time by merging the two linked lists of arrays. Deletions can be performed lazily, shrinking the array when a constant fraction of the entries have been deleted.

We analyze the number of pointers required of this method and the resulting improved bounds on insertion and key change in \wref{sec:qbounds}. If we choose not to take advantage of this directly, we can alternatively replace standard linked lists with linked list/array hybrids such as unrolled linked lists~\cite{Shao94}, which will likely outperform standard linked lists in practice.

\subsection{Insertion}
\label{sec:insertalg}

Insertion of an element $e=(k,v)$ can be succinctly described as follows. We first determine the gap $\Delta_i$ such that $k \in \Delta_i$, according to the data structure of \wref{lem:gapds}. We then binary search the $O(\log |\Delta_i|)$ intervals (by maintaining ``router'' keys separating the intervals) within $\Delta_i$ to find the interval $\mathcal{I}_{i,j}$ such that $k \in \mathcal{I}_{i,j}$. We increase the size of $\Delta_i$ by one in the gap data structure.

\begin{remark}[A single data structure]
\label{rem:work-over-intervals}
The attentive reader may wonder why we must first perform a binary search for gap $\Delta_i$ and then perform another binary search for interval $\mathcal{I}_{i,j}$ within $\Delta_i$. It seems likely these two binary searches can be compressed into one, and indeed, this intuition is correct. If preferred, we can use the data structure of \wref{lem:gapds} directly on intervals within gaps, so that weight $|\Delta_i|$ is evenly distributed over intervals $\mathcal{I}_{i,1}, \mathcal{I}_{i,2}, \ldots, \mathcal{I}_{i,\ell_i}$. (Alternatively, assigning weight $|\Delta_i|/\ell_i + |\mathcal{I}_{i,j}|$ to interval $\mathcal{I}_{i,j}$ can provide better runtimes in average case settings.) Unfortunately, doing so means only an $O(\log n)$ time change-key operation can be supported (unless the data structure is augmented further), and (small) weight changes must be performed on the full set of intervals within gap $\Delta_i$ on insertion and deletion. While such a data structure is possible, we find the current presentation more elegant and simpler to implement.
\end{remark}

\begin{remark}[Lazy insert]
One seemingly-obvious way to improve insertion complexity, improving perhaps either of the first two disadvantages listed in \wref{sec:con}, is to \emph{insert lazily}. That is, instead of performing a full insert of $e = (k,v)$ through the gap data structure and then again through the interval data structure, we keep a buffer at each node of the respective BSTs with all the elements that require processing at a later time. While this can improve overall time complexity on some simple operation sequences, it seems difficult to make this strategy efficient overall, when insertions, deletions and queries can be mixed arbitrarily. 

So while improving either of the two disadvantages listed in \wref{sec:con} (and indeed, an improvement in one may imply an improvement in the other) would likely utilize aspects of lazy insertion, we do not currently see a way to achieve this by maintaining buffers on nodes of the BSTs we use.
\end{remark}

\subsection{Query}
\label{sec:queryalg}

To answer a query with associated rank $r$, we proceed as follows. We again determine the gap~$\Delta_i$ such that $r \in \Delta_i$ according to the process described in \wref{def:rankbasedquery} on the data structure of \wref{lem:gapds}. %
While we could now choose to rebalance the intervals of $\Delta_i$ via \wref{rule:merge}, our analysis will not require application of \wref{rule:merge} until the \emph{end} of the query procedure.
We recurse into the interval $\mathcal{I}_{i,j}$ such that $r \in \mathcal{I}_{i,j}$, again using the process described in \wref{def:rankbasedquery} on the intervals of~$\Delta_i$ (this may use aggregate information stored in the data structure for intervals).

We proceed to process $\mathcal{I}_{i,j}$ by answering the query on $\mathcal{I}_{i,j}$ and replacing interval $\mathcal{I}_{i,j}$ with smaller intervals. 
First,
we partition $\mathcal{I}_{i,j}$ into sets $L$ and $R$, such that all elements in $L$ are less than or equal to all elements in $R$ and there are $r$ elements in the entire data structure which are either in $L$ or in an interval or gap left of $L$. 
This can typically be done in $O(|\mathcal{I}_{i,j}|)$ time using the result of the query itself; otherwise, linear-time selection suffices~\cite{Blum73}.

We further partition $L$ into two sets of equal size $L_l$ and $L_r$, again using linear-time selection, such that all elements in $L_l$ are smaller than or equal to elements in $L_r$; if $|L|$ is odd, we give the extra element to $L_l$ (unsurprisingly, this is not important). 
We then apply the same procedure \emph{one more time} to $L_r$, again splitting into equal-sized intervals.
Recursing further is not necessary.
We do the same, reflected, for set $R$; 
after a total of 5 partitioning steps the interval splitting terminates. 
An example is shown in \wref{fig:split}.

\begin{figure}[htbp]
	\def\intsep{0.08}
	\newcommand\drawint[4][1]{
		\draw[semithick,|-|] (#2,0) -- 
		node[below=0pt,font=\scriptsize] {\scalebox{#1}[1]{$#4$}} 
		++(#3-\intsep,0) ;
	}
	\scalebox{1.2}{%
		\begin{tikzpicture}[scale=.4]
		\foreach \x in {-.1,6,19.1-\intsep} {
			\draw[densely dotted,thin] (\x,0) -- (\x,-6) ;
		}
		
		\drawint{-.1}{19.2}{|\mathcal I_{i,j}|=19}
		\node[font=\small] at (3*19/4,1) {interval $\mathcal I_{i,j}$} ;
		\draw[thick,gray] (6,-.5) -- ++ (0,1) 
		node[above,black,font=\scriptsize,align=center] 
		{query\\rank\\$r=6$} ;
		\begin{scope}[shift={(0,-5)}]
		\begin{scope}[shift={(-0.1,0)}]
		\drawint{0}{3}{3}
		\drawint{3}{2}{2}
		\drawint{5}{1}{1}
		\end{scope}
		\begin{scope}[shift={(0.1,0)}]
		\drawint{6}{3}{3}
		\drawint{9}{3}{3}
		\drawint{12}{7}{7}
		\end{scope}
		\end{scope}
		\node[scale=2,rotate=-90] at (19/2,-2.75) {$\Rightarrow$};
		
		\foreach \f/\t/\l in {0/6/L,6/19/R} {
			\draw[decoration={brace,amplitude=5pt,mirror},decorate]
			(\f+.1,-6.75) -- node[below=5pt,font=\scriptsize]{$\l$} ++(\t-\f-.2,0);
		}
		\end{tikzpicture}%
	}
	\caption{An interval $\mathcal{I}_{i,j}$ is split and replaced with a set of intervals.}
	\label{fig:split}
\end{figure}
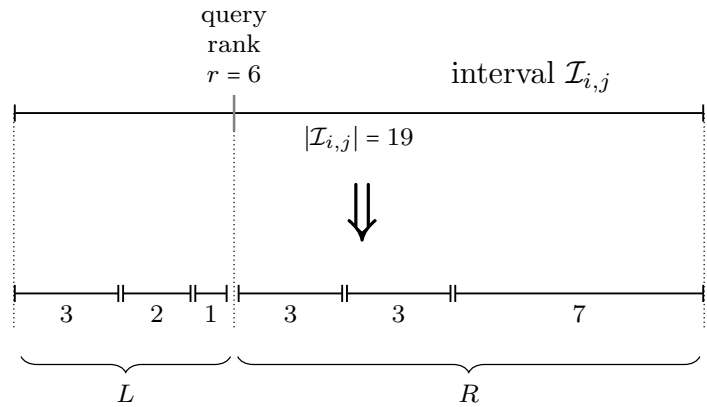

\begin{remark}[Variants of interval replacement]
There is some flexibility in designing this interval-replacement procedure;
the critical property needed for our result is the following;
(details of which will become clear in \wref{sec:queryanalysis}):
(1) It yields at most $O(\log |\Delta_i|)$ intervals in gap $\Delta_i$ (typically by application of \wref{rule:merge}), (2) it satisfies an invariant involving a credit system~-- \wtpref{inv:credits}~-- and 
(3) splitting takes time~$O(|\mathcal{I}_{i,j}|)$.
In \wref{sec:random}, we show that exact median selection (when splitting $L$, $L_r$, $R$, and $R_l$) 
can be replaced with pivoting on a randomly chosen element. 
On a set of $n$ elements, this requires only $n$ comparisons 
instead of the at least $1.5n$ required by median-finding in expectation~\cite{Cunto89},
and it is substantially faster in practice.
\end{remark}

After splitting the interval $\mathcal I_{i,j}$ as described above,
we answer the query itself and update the gap and interval data structures as follows. We create two new gaps $\Delta'_i$ and $\Delta'_{i+1}$ out of the intervals of gap $\Delta_i$ including those created from sets $L$ and $R$. Intervals that fall left of the query rank $r$ are placed in gap $\Delta'_i$, and intervals that fall right of the query rank $r$ are placed in gap $\Delta'_{i+1}$. We update the data structure of \wref{lem:gapds} with the addition of gaps $\Delta'_i$ and $\Delta'_{i+1}$ and removal of gap~$\Delta_i$. Finally, we apply \wref{rule:merge} to gaps $\Delta'_i$ and~$\Delta'_{i+1}$.

\subsection{Deletion}

To delete an element $e=(k,v)$ pointed to by a given pointer \texttt{ptr}, we first remove $e$ from the interval $\mathcal{I}_{i,j}$ such that $k \in \mathcal{I}_{i,j}$. 
If $e$ was the only element in $\mathcal{I}_{i,j}$, we remove interval $\mathcal{I}_{i,j}$ from gap $\Delta_i$ (we can do so lazily, when \wref{rule:merge} is next run on gap $\Delta_i$). 
Then we decrement $\Delta_i$ in the gap data structure of \wref{lem:gapds};
if that leaves an empty gap, we remove $\Delta_i$ from the gap data structure.

\subsection{Change-Key}
\label{sec:changekeyalg}

The change-key operation can be performed as follows. Suppose we wish to change the key of element $e=(k,v)$, given by pointer \texttt{ptr}, to $k'$, and that $e$ currently resides in interval $\mathcal{I}_{i,j}$ in gap $\Delta_i$. We first check if $k'$ falls in $\Delta_i$ or if $e$ should be moved to a different gap. If the latter, we can do so as in deletion of $e$ and re-insertion of $(k',v)$. If the former, we first remove $e$ from $\mathcal{I}_{i,j}$. If necessary, we (lazily) delete $\mathcal{I}_{i,j}$ from $\Delta_i$ if $\mathcal{I}_{i,j}$ now has no elements. We then binary search the $O(\log |\Delta_i|)$ intervals of $\Delta_i$ and place $e$ into the new interval in which it belongs.

Note that although this operation can be performed to change the key of $e$ to anything, \wref{thm:main} only guarantees runtimes faster than $O(\log n)$ when $e$ moves closer to its nearest query rank within gap $\Delta_i$. Efficient runtimes are indeed possible in a variety of circumstances; this is explored in more detail in \wref{sec:changekeyanalysis}.

\section{Analysis}
\label{sec:analysis}

We use an amortized analysis~\cite{Tarjan85}. We will use a potential function with a built-in credit system.
Recall that our desired insertion complexity is about $O(\log(n/|\Delta_i|))$ time. On a query that splits a gap into two gaps of size $x$ and $cx$, we attempt to do so in (amortized) $O(\log n + x \log c)$ time. 
We require several definitions before we may proceed.

We distinguish between $0$-sided, $1$-sided, and $2$-sided gaps. A $2$-sided gap is a gap $\Delta_i$ such that queries have been performed on either side of $\Delta_i$; thus, intervals in $\Delta_i$ are split into intervals on the left side and intervals on the right side. This is the typical case. A $1$-sided gap $\Delta_i$ is such that queries have only been performed on one side of the gap; thus, intervals are all on the side towards the query rank in $\Delta_i$. There can be at most two $1$-sided gaps at any point in time. In the priority queue case, there is a single $1$-sided gap. The final category is a $0$-sided gap; when the data structure has performed no queries, all elements are represented in a single interval in a single $0$-sided gap.

We now give the following functional definitions.
\begin{align*}
	c(\mathcal{I}_{i,j}) &\;\ce\; \text{\# of \emph{credits} associated with interval $\mathcal{I}_{i,j}$.}\\
	o(\mathcal{I}_{i,j}) &\;\ce\; \text{\# of elements \emph{outside} $\mathcal{I}_{i,j}$ in $\Delta_i$, i.\,e.,}\\ 
	&\phantom{{}\;\ce\;{}}\text{\# of elements in $\Delta_i$ that are left (right) of $\mathcal{I}_{i,j}$ if $\mathcal{I}_{i,j}$ is on the left (right) side.}\\
	M &\;\ce\; \text{total \# of elements in $0$-sided or $1$-sided gaps.}
\end{align*}

As previously mentioned, intervals are defined to be on either the left or right side according to \wpref{rule:left-right}. For an interval $\mathcal{I}_{i,j}$ in a $2$-sided gap $\Delta_i$, $o(\mathcal{I}_{i,j})$ hence is the minimum number of elements either to the left (less than) or to the right (greater than) $\mathcal{I}_{i,j}$ in gap $\Delta_i$.

The rules for assigning credits are as follows: 
A newly created interval has no credits associated with it.
During a merge, the credits associated with both intervals involved in the merge may be discarded; they are not needed.
When an interval $\mathcal I_{i,j}$ is split upon a query, it is destroyed and new intervals 
(with no credits) are created from it; by destroying $\mathcal I_{i,j}$, the $c(\mathcal I_{i,j})$ credits associated with it are released.

We use the following potential function:
\[
\Phi \;=\; 10M \;+\; 4\sum_{\mathclap{\substack{1 \leq i \leq m,\\1 \leq j \leq \ell_i}}} \: 
	c(\mathcal{I}_{i,j}).
\]
Credits accumulated when an operation is cheaper than 
its amortized bound increase $\Phi$; 
in this way, we use credits to pay for work that will need to be performed in the future. 
We do so by maintaining the following invariant:
\begin{enumerate}[label=(\Alph*),start=3,font=\bfseries]
\item \label{inv:credits}
	\textbf{Credit invariant:}
	Let $\mathcal{I}_{i,j}$ be an interval. 
	Then $|\mathcal{I}_{i,j}| \le c(\mathcal{I}_{i,j}) + o(\mathcal{I}_{i,j})$.
\end{enumerate}
\begin{remark}[Intuition behind \wref{inv:credits}]
	The intuition behind \wref{inv:credits} is that the cost of splitting $\mathcal{I}_{i,j}$ is entirely paid for by the credits associated with $\mathcal{I}_{i,j}$ and by outside elements,
\ie, either released potential or by the distance to previous queries causing a corresponding increase in $B$.
	The intervals constructed from the elements of $\mathcal{I}_{i,j}$ are constructed in such a way that they satisfy \wref{inv:credits} at cost a constant fraction of the cost of splitting $\mathcal{I}_{i,j}$.
\end{remark}
\begin{remark}[Alternative potential function]
	It is possible to remove the credits in our potential function and \wref{inv:credits} and instead use the potential function
\[
\Phi \;=\; 10M \;+\; 4\sum_{\mathclap{\substack{1 \leq i \leq m,\\1 \leq j \leq \ell_i}}} \: 
\max(|\mathcal{I}_{i,j}| - o(\mathcal{I}_{i,j}), 0).
\]
We opt for the current method as we believe it is easier to work with.
\end{remark}
Observe that before any elements are inserted, $\Phi = 0$, and we have a single $0$-sided gap with one interval containing no elements. Thus \wref{inv:credits} is vacuously true. We proceed with an amortized analysis of the operations. For our amortization arguments, we assume the potential function to be adjusted to the largest constant in the $O(\cdot)$ notation necessary for the complexity of our operations. 
In the interest of legibility, we will drop this constant and compare outside of $O(\cdot)$ notation, as is standard in amortized complexity analysis.

\subsection{Insertion}
\label{sec:insertanalysis}

Insertion of element $e=(k,v)$ can be analyzed as follows. As stated in \wref{lem:gapds}, we pay $O(\log(n/|\Delta_i|))$ time to locate the gap $\Delta_i$ that $e$ belongs into. We adjust the size of $\Delta_i$ and $n$ by one in the data structure of \wref{lem:gapds} in $O(\log(n/|\Delta_i|))$ time. By \wref{lem:numintervals}, there are $O(\log |\Delta_i|)$ intervals in gap $\Delta_i$, and so we spend $O(\log \log |\Delta_i|)$ time to do a binary search to find the interval $\mathcal{I}_{i,j}$ that $e$ belongs into. We increase the size of $\mathcal{I}_{i,j}$ by one and add one credit to $\mathcal{I}_{i,j}$ to ensure \wref{inv:credits}. Thus the total amortized cost of insertion\footnote{Note that although $\log \log |\Delta_i|$ can be $o(\log \log n)$, there is no difference between $O(\sum_{i=1}^{q+1} |\Delta_i| (\log (n/|\Delta_i| )+ \log \log n))$ and $O(\sum_{i=1}^{q+1} |\Delta_i| (\log(n/|\Delta_i|) + \log \log |\Delta_i|)$. When the $\log \log n$ term in a gap dominates, $|\Delta_i| = \Omega(n/ \log n)$, so $\log \log n = \Theta(\log \log |\Delta_i|)$.} (up to constant factors) is $\log(n/|\Delta_i|)+\log \log |\Delta_i| + 4+10 = O(\log(n/|\Delta_i|) + \log \log |\Delta_i|)$. Note that if the data structure for \wref{lem:gapds} supports operations in worst-case time, insertion complexity is also worst-case.
We show in \wref{sec:qbounds} that the bound $O(\log q)$ also holds.

We use the following lemma to show that \wref{inv:credits} holds on insertion.

\begin{lemma}[Insert maintains \wref{inv:credits}]
\label{lem:ruleAinvariant}
Updating side designations according to \wref{rule:left-right} after insertions preserves \wref{inv:credits}.
\end{lemma}
\begin{proof}
The insertion of additional elements may cause an interval $\mathcal{I}_{i,j'}$ in the middle of $\Delta_i$ to change sides. This occurs exactly when the number of elements on one side exceeds the number of elements on the other side. However, before this insertion occurred, \wref{inv:credits} held with an equal number of elements on both sides of $\mathcal{I}_{i,j'}$. Since we do not change the number of credits associated with $\mathcal{I}_{i,j'}$, in effect, $o(\mathcal{I}_{i,j'})$ just changes which side it refers to, monotonically increasing through all insertions. It follows \wref{inv:credits} holds according to redesignations via \wref{rule:left-right} after insertions.
\end{proof}

\wref{inv:credits} then holds on insertion due to \wref{lem:ruleAinvariant} and since $o(\mathcal{I}_{i,j'})$ only possibly increases for any interval $\mathcal{I}_{i,j'}$, $j' \neq j$, with $|\mathcal{I}_{i,j'}|$ remaining the same; recall that an extra credit was added to interval $\mathcal{I}_{i,j}$ to accommodate the increase in $|\mathcal{I}_{i,j}|$ by one.

Note that from an implementation standpoint, no work need be done for intervals $\mathcal{I}_{i,j'}$ on insertion, even if they change sides. Any readjustment can be delayed until the following query in gap $\Delta_i$.

\subsection{Query}
\label{sec:queryanalysis}

We now proceed with the analysis of a query. We split the analysis into several sections. We first assume the gap $\Delta_i$ in which the query falls is a $2$-sided gap. We show \wref{inv:credits} implies we can pay for the current query. We then show how to ensure \wref{inv:credits} holds after the query. Finally, we make the necessary adjustments for the analysis of queries in $0$-sided and $1$-sided gaps. Recall that our complexity goal to split a gap into gaps of size $x$ and $cx$ ($c \geq 1$)  is $O(\log n + x \log c)$ amortized time.

\subsubsection{Current Query}
\label{sec:current}

For the moment, we assume the gap in which the query rank $r$ satisfies $r \in \Delta_i$ is a $2$-sided gap. Further, assume the query rank $r$ falls left of the median of gap $\Delta_i$, so that the resulting gaps are a gap $\Delta'_i$ of size $x$ and a gap $\Delta'_{i+1}$ of size $cx$ ($c \geq 1$). A picture is given in \wref{fig:query}. The case of query rank $r$ falling right of the median of $\Delta_i$ is symmetric.

\begin{figure}[htbp]
	\def\intsep{0.2}
	\newcommand\drawint[2]{
		\draw[intline,|-|] (#1,0) -- ++(#2-\intsep,0) ;
	}
	\scalebox{1.4}{%
	\begin{tikzpicture}[xscale=.12,yscale=.18]
	\scriptsize
		\tikzset{intline/.style={thin}}
		\drawint{0}{1}
		\drawint{1}{3}
		\drawint{4}{5}
		\drawint{9}{8}
		\drawint{17}{8}
		\drawint{25}{10}
		\drawint{35}{15}
		\drawint{50}{21}
		\drawint{71}{13}
		\drawint{84}{4}
		\drawint{88}{2}
		\drawint{90}{1}
		\drawint{91}{1}
		
		\node[scale=.7] at (9+4,1) {$\mathcal I_{i,j}$};
		
		\draw[thick,draw=gray] (15,1) -- ++(0,-4) node[below,scale=.7] {query};
		
		\foreach \f/\t/\y/\l in {%
				0/92-\intsep/9/{gap $\Delta_i$},%
				0/15-\intsep/4.5/{gap $\Delta_i'$\\$|\Delta_i'| = x$},%
				15/92-\intsep/4.5/{gap $\Delta_{i+1}'$\\$|\Delta_{i+1}'|=cx$}%
		} {
			\draw[decoration={brace,amplitude=5pt},decorate] (\f+.2,\y) -- 
			node[above=5pt,scale=.7,align=center] {\l} ++(\t-\f-.4,0);
		}
		
		\draw[decoration={brace,mirror,amplitude=2.5pt},decorate]
			(0+.2,-2) -- node[below=2pt,scale=.7,align=center] {
				intervals outside\\ of $\mathcal I_{i,j}$ in $\Delta_i$ 
				will\\ move to gap $\Delta_i'$
			} ++(9-.4,0) ;
		\draw[decoration={brace,mirror,amplitude=2.5pt},decorate]
			(17+.2,-2) -- node[below=2pt,scale=.7,align=center] {
				intervals on \\same side of $\mathcal I_{i,j}$,\\
				will move to $\Delta_{i+1}'$
			} ++(50-17-.4,0) ;
		\draw[stealth-,very thin] (35+15/2,0.2) to[out=70,in=270] ++(60:1.8) 
				node[above,align=center,scale=.6] {
					last interval on\\ left side of $\Delta_i$} ;
	\end{tikzpicture}
	}
  \caption{A query that splits $\mathcal{I}_{i,j}$ in gap $\Delta_i$.}
  \label{fig:query}
\end{figure}
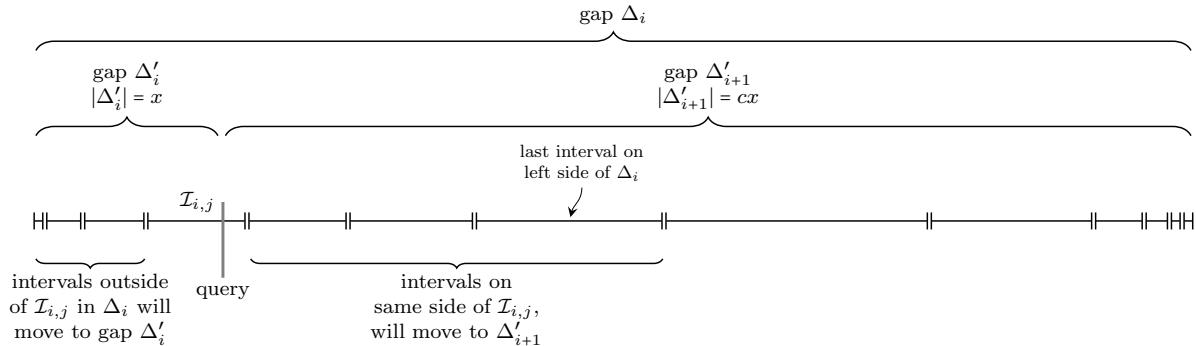

It takes $O(\log(n/|\Delta_i|)) = O(\log n)$ time via the data structure of \wref{lem:gapds} to find the gap~$\Delta_i$. We then find the interval $\mathcal{I}_{i,j}$ such that $r \in \mathcal{I}_{i,j}$. By \wref{def:rankbasedquery}, answering the query on the set of unsorted elements $\mathcal{I}_{i,j}$ can be done in $O(|\mathcal{I}_{i,j}|)$ time. Splitting interval $\mathcal{I}_{i,j}$ as described in \wref{sec:queryalg} can also be done in $O(|\mathcal{I}_{i,j}|)$ time.

Updating the data structure of \wref{lem:gapds} with the addition of gaps $\Delta'_i$ and $\Delta'_{i+1}$ and removal of gap $\Delta_i$ can be done in $O(\log n)$ time. Similarly, the total number of intervals created from the current query is no more than 6, and no more than $O(\log |\Delta_i|)$ intervals existed in gap $\Delta_i$ prior to the query, again by \wref{lem:numintervals}. Thus, applying \wref{rule:merge} to gaps $\Delta'_i$ and $\Delta'_{i+1}$ after the query takes no more than $O(\log |\Delta_i|) = O(\log n)$ time, because merging two intervals can be done in $O(1)$ time. 

We next show that merging of intervals according to \wref{rule:merge} will preserve \wref{inv:credits}.
\begin{lemma}[Merge maintains \ref{inv:credits}]
	Suppose interval $\mathcal{I}_{i,j}$ is merged into interval $\mathcal{I}_{i,j'}$ (note $j' = j+1$ if $\mathcal{I}_{i,j}$ is on the left side and $j' = j-1$ if $\mathcal{I}_{i,j}$ is on the right side), according to \wref{rule:merge}. Then the interval $\mathcal{I}_{i,j'}$ after the merge satisfies \wref{inv:credits}.
	\label{ruleBinvariant}
\end{lemma}
\begin{proof}
	Suppose interval $\mathcal{I}_{i,j}$ is merged into interval $\mathcal{I}_{i,j'}$ according to \wref{rule:merge}. Then $o(\mathcal{I}_{i,j}) \geq |\mathcal{I}_{i,j}| + |\mathcal{I}_{i,j'}|$. This implies that after the merge,  $o(\mathcal{I}_{i,j'}) \geq |\mathcal{I}_{i,j'}|$, since elements outside the merged interval $\mathcal{I}_{i,j'}$ are outside both of the original intervals. Thus $\mathcal{I}_{i,j'}$ satisfies \wref{inv:credits} without any credits.
\end{proof}

In total, we pay $O(\log n + |\mathcal{I}_{i,j}|)$ actual time. As the $O(\log n)$ component is consistent with the $O(\log n)$ term in our desired query complexity, let us focus on the $O(|\mathcal{I}_{i,j}|)$ term. We have the following.

\begin{lemma}[Amortized splitting cost]
\label{lem:cancel}
Consider a query which falls in interval $\mathcal{I}_{i,j}$ and splits gap $\Delta_i$ into gaps of size $x$ and $cx$. Then $|\mathcal{I}_{i,j}| - c(\mathcal{I}_{i,j}) \leq x$.
\end{lemma}
\begin{proof}
By \wref{inv:credits}, $|\mathcal{I}_{i,j}| - c(\mathcal{I}_{i,j}) \leq o(\mathcal{I}_{i,j})$. 
Now, since $\Delta_i$ is a $2$-sided gap, $o(\mathcal{I}_{i,j})$ is the lesser of the number of elements left or right of $\mathcal{I}_{i,j}$. Since the query rank $r$ satisfies $r \in \mathcal{I}_{i,j}$, this implies $o(\mathcal{I}_{i,j}) \leq x$ (See \wref{fig:query} for a visual depiction).
\end{proof}

We can apply amortized analysis with \wref{lem:cancel} as follows. Interval $\mathcal{I}_{i,j}$ is destroyed and intervals that are built from its contents have no credits. Thus, $4c(\mathcal{I}_{i,j})$ units of potential are released. By applying \wref{lem:cancel}, we can use $c(\mathcal{I}_{i,j})$ units of this released potential to bound the cost $|\mathcal{I}_{i,j}|$ with $x$. This gives an amortized cost thus far of $\log n + x - 3c(\mathcal{I}_{i,j})$. We will use the extra $3c(\mathcal{I}_{i,j})$ units of potential in the following section, ensuring \wref{inv:credits} holds for future operations.

\subsubsection[Ensuring Invariant (C)]{Ensuring \wref{inv:credits}}
\label{sec:ensure}

We must ensure \wref{inv:credits} holds on all intervals in gaps $\Delta'_i$ and $\Delta'_{i+1}$. Again, we will suppose $\Delta'_i$ is the smaller gap of the two, so that $\Delta'_i$ has $x$ elements and $\Delta'_{i+1}$ has $cx$ elements; the other case is symmetric.

Let us first consider gap $\Delta'_i$. This gap contains intervals from $\Delta_i$ outside of $\mathcal{I}_{i,j}$ as well as intervals made from the elements of $\mathcal{I}_{i,j}$. Observe (cf.\ \wref{fig:query}) that gap $\Delta'_i$ has in total $x$ elements. Therefore, we can trivially ensure \wref{inv:credits} holds by adding enough credits to each interval of $\Delta'_i$ to make it so, at total amortized cost at most $4x$. Let us do this after applying \wref{rule:merge} to $\Delta'_i$, so it is balanced and satisfies \wref{inv:credits}.

We now consider gap $\Delta'_{i+1}$ after rebalancing according to \wref{rule:merge}. The application of \wref{rule:merge} after the query may cause some intervals to change sides towards the query rank $r$ and subsequently merge. Intervals created from $\mathcal{I}_{i,j}$ may also merge (this may be because \wref{rule:merge} was applied lazily or even because the largest interval created from $\mathcal{I}_{i,j}$ may be on the opposite side of the rest of the intervals created from interval $\mathcal{I}_{i,j}$).
In total, the intervals of $\Delta'_{i+1}$ fall into four distinct categories. Recall that when we apply \wref{rule:merge}, we merge an interval $\mathcal{I}_{i,j'}$ into interval $\mathcal{I}_{i,j''}$, so we assume the identity of the merged interval as $\mathcal{I}_{i,j''}$, and interval $\mathcal{I}_{i,j'}$ ceases to exist.

We call the four categories $A$, $B$, $C$, and $D$, and show how to ensure \wref{inv:credits} on each of them. 
Category $A$ are intervals that are created from $\mathcal{I}_{i,j}$ that fall on the side of the query rank $r$ so as to become $\Delta'_{i+1}$ intervals after the query. Category $B$ are intervals on the same side as $\mathcal{I}_{i,j}$ before the query which were located inward from $\mathcal{I}_{i,j}$ in $\Delta_i$. Category $C$ are intervals that were on the opposite side of interval $\mathcal{I}_{i,j}$ before the query, but now switch sides due to the removal of the gap $\Delta'_i$. Finally, category $D$ are intervals that lie on the opposite side of interval $\mathcal{I}_{i,j}$ both before and after the query.
A picture is given in \wref{fig:newgap}.

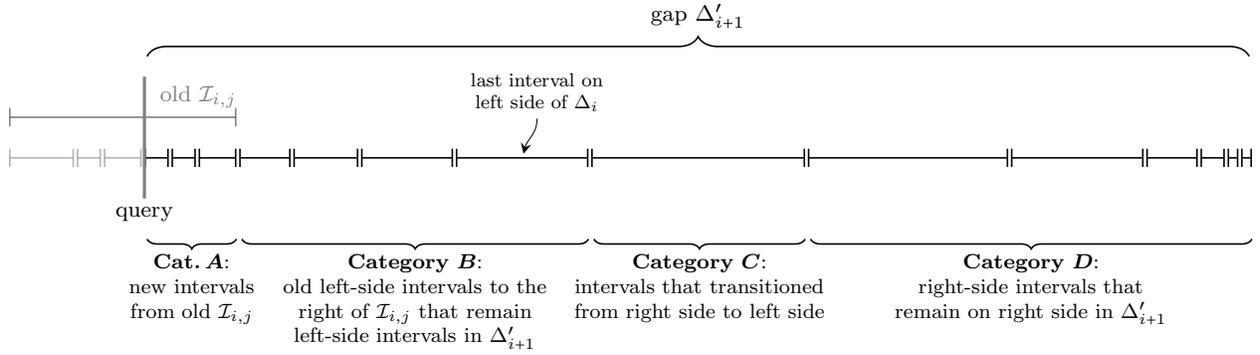
\begin{figure}[htbp]
	\def\intsep{0.2}
	\newcommand\drawint[2]{
		\draw[intline,|-|] (#1,0) -- ++(#2-\intsep,0) ;
	}
	\scalebox{1.5}{%
	\begin{tikzpicture}[xscale=.12,yscale=.18]
	\scriptsize
		\tikzset{intline/.style={thin}}
		\drawint{17}{4}
		\drawint{21}{5}
		\drawint{26}{7}
		\drawint{33}{10}
		\drawint{43}{16}
		\drawint{59}{15}
		\drawint{74}{10}
		\drawint{84}{4}
		\drawint{88}{2}
		\drawint{90}{1}
		\drawint{91}{1}
		\begin{scope}[shift={(0,2)},opacity=.5]
			\drawint{0}{8+9}
			\node[scale=.7] at (9+5,1) {old $\mathcal I_{i,j}$};
		\end{scope}
		\begin{scope}[opacity=.3]

			\drawint{0}{5}
			\drawint{5}{2}
			\drawint{7}{3}

		\end{scope}

		\drawint{10}{2}
		\drawint{12}{2}
		\drawint{14}{3}
		
		\draw[thick,draw=gray] (10,4) -- ++(0,-6) node[below,scale=.7] {query};
		
		\foreach \f/\t/\y/\l in {%
				10/92-\intsep/5/{gap $\Delta_{i+1}'$}%
		} {
			\draw[decoration={brace,amplitude=5pt},decorate] (\f+.2,\y) -- 
			node[above=5pt,scale=.7,align=center] {\l} ++(\t-\f-.4,0);
		}

		\foreach \f/\t/\l in {%
			10/17/{\textbf{\boldmath Cat.\,$A$}:\\new intervals \\from old $\mathcal I_{i,j}$},%
			17/43/{\textbf{\boldmath Category $B$}:\\old left-side intervals to the \\right of $\mathcal I_{i,j}$ that remain \\left-side intervals in $\Delta_{i+1}'$},%
			43/59/{\textbf{\boldmath Category $C$}:\\intervals that transitioned \\from right side to left side},%
			59/92/{\textbf{\boldmath Category $D$}:\\right-side intervals that \\remain on right side in $\Delta_{i+1}'$}%
		} {
			\draw[decoration={brace,mirror,amplitude=2.5pt},decorate]
				(\f+.2,-4) -- node[below=2pt,scale=.65,align=center] {\l} 
				++(\t-\f-.4,0) ;
		}
		\draw[stealth-,very thin] (33+5,0.2) to[out=70,in=270] ++(60:2) 
				node[above,align=center,scale=.6] 
						{last interval on\\left side of $\Delta_i$} ;
	\end{tikzpicture}
	}
  \caption{Gap $\Delta'_{i+1}$ after query within interval $\mathcal{I}_{i,j}$ of $\Delta_i$. The picture assumes $\mathcal{I}_{i,j}$ was a left-side interval.}
  \label{fig:newgap}
\end{figure}

 We proceed with ensuring \wref{inv:credits} on each category.

\begin{itemize}
  \item \textbf{\boldmath Category $D$}: Category $D$ intervals are easiest. These intervals are not affected by the query and thus still satisfy \wref{inv:credits} with no additional cost.
  
  \item \textbf{\boldmath Category $A$}: Now consider category $A$ intervals. 
  Three such intervals, $\mathcal I'_{i+1,1}$, $\mathcal I'_{i+1,2}$, and $\mathcal I'_{i+1,3}$, are created in the query algorithm stated in \wref{sec:queryalg}. The leftmost and middlemost intervals, $\mathcal I'_{i+1,1}$ and $\mathcal I'_{i+1,2}$, have size $\frac{1}{4}|R| \pm 1$, and the rightmost interval $\mathcal I'_{i+1,3}$ has size $\frac{1}{2}|R| \pm 1$ (the $\pm1$ addresses the case that $|R|$ is not divisible by $4$). 
  
  Up to one element, $\mathcal I'_{i+1,2}$ has at least as many elements outside of it as within it. Thus after giving it one credit, $\mathcal I'_{i+1,2}$ satisfies \wref{inv:credits}. Similarly, $\mathcal I'_{i+1,3}$ will remain on the same side in most cases, and thus will also have enough elements outside it from the other two intervals (potentially after giving it one credit, too).
  But we always have to assign credits to $\mathcal I'_{i+1,1}$.
  Moreover, if interval $\mathcal{I}_{i,j}$ was very large, then $\mathcal I'_{i+1,3}$ may actually switch sides in the new gap $\Delta'_{i+1}$. 
  
  In the worst case, we will require credits to satisfy \wref{inv:credits} on both $\mathcal I'_{i+1,1}$ and $\mathcal I'_{i+1,3}$. As their sizes total $\frac{3}{4}|R| + O(1)$, at $4$ units of potential per credit the amortized cost to do so is no more than $3|\mathcal{I}_{i,j}| + O(1)$. We can use the extra $3c(\mathcal{I}_{i,j})$ units of potential saved from \wref{sec:current} to pay for this operation. By applying \wref{lem:cancel} again, we can bound $3|\mathcal{I}_{i,j}| - 3c(\mathcal{I}_{i,j})$ with $3x$, bringing the amortized cost of satisfying \wref{inv:credits} on category $A$ intervals to $O(x)$.

  \item \textbf{\boldmath Categories $B$ and $C$}: We'll handle category $B$ and $C$ intervals together. First observe that since $x$ elements were removed with the query, we can bound the number of credits necessary to satisfy \wref{inv:credits} on a single interval in category $B$ or $C$ with either $x$ or the size of that interval. For category $C$ intervals, this follows because they had more elements on their left side prior to the query, thus upon switching sides after the query, $x$ credits will suffice to satisfy \wref{inv:credits}, similarly to the proof of \wref{lem:ruleAinvariant}. In the new gap $\Delta'_{i+1}$, let $j'$ be the smallest index such that $|\mathcal{I}_{i+1,j'}| \geq x$. We will handle category $B$ and $C$ intervals left of $\mathcal{I}_{i+1,j'}$ and right of $\mathcal{I}_{i+1,j'}$ differently.
  
  Let us first consider category $B$ and $C$ intervals left of interval $\mathcal{I}_{i+1,j'}$. All such intervals have size less than $x$. If there are less than two such intervals, we may apply $x$ credits to each to ensure \wref{inv:credits} at total cost $O(x)$. Otherwise, consider intervals $\mathcal{I}_{i+1,j'-2}$ and $\mathcal{I}_{i+1,j'-1}$. Due to application of \wref{rule:merge} after the query, intervals $\mathcal{I}_{i+1,j'-2}$ and $\mathcal{I}_{i+1,j'-1}$ make up more than half of the total number of elements left of interval $\mathcal{I}_{i+1,j'}$. Since $|\mathcal{I}_{i+1,j'-2}| < x$ and $|\mathcal{I}_{i+1,j'-1}| < x$, it follows there are no more than $4x$ elements located in intervals left of interval $\mathcal{I}_{i+1,j'}$ in gap $\Delta'_{i+1}$. For each such interval, we add at most the size of the interval in credits so that \wref{inv:credits} holds on all intervals left of $\mathcal{I}_{i+1,j'}$ in gap $\Delta'_{i+1}$. The total cost is $O(x)$.
  
  Now consider intervals right of $\mathcal{I}_{i+1,j'}$. If there are less than two such intervals, we may apply $x$ credits to each to ensure \wref{inv:credits} at total cost $O(x)$. Otherwise, consider intervals $\mathcal{I}_{i+1,j'+1}$ and $\mathcal{I}_{i+1,j'+2}$. By \wref{rule:merge} after the query, $|\mathcal{I}_{i+1,j'+1}| + |\mathcal{I}_{i+1,j'+2}| > x$, since interval $\mathcal{I}_{i+1,j'}$ is outside intervals $\mathcal{I}_{i+1,j'+1}$ and $\mathcal{I}_{i+1,j'+2}$ and $|\mathcal{I}_{i+1,j'}| \geq x$ by choice of $j'$. Similarly, if such intervals are category $B$ or $C$ intervals, then $|\mathcal{I}_{i+1,j'+3}| + |\mathcal{I}_{i+1,j'+4}| > 2x$ and $|\mathcal{I}_{i+1,j'+5}| + |\mathcal{I}_{i+1,j'+6}| > 4x$. In general, $|\mathcal{I}_{i+1,j'+2k-1}| + |\mathcal{I}_{i+1,j'+2k}| > 2^{k-1}x$ for any $k$ where intervals $\mathcal{I}_{i+1,j'+2k-1}$ and $\mathcal{I}_{i+1,j'+2k}$ are category $B$ or $C$ intervals. Since there are $cx$ total elements in gap $\Delta'_{i+1}$, it follows the number of category $B$ and $C$ intervals right of $\mathcal{I}_{i+1,j'}$ is $O(\log c)$. We may then apply $x$ credits to all such intervals and interval $\mathcal{I}_{i+1,j'}$ for a total cost of $O(x \log c)$. 

\end{itemize}

Altogether, we can ensure \wref{inv:credits} for future iterations at total $O(x \log c)$ amortized cost.

\subsubsection{0-Sided and 1-Sided Gaps}
\label{sec:onesided}

We proceed with a generalization of the previous two sections for when the gap $\Delta_i$ in which the query falls is a $0$-sided or $1$-sided gap. If gap $\Delta_i$ is $0$-sided, we spend $O(n)$ time to answer the query, according to \wref{def:rankbasedquery} on a set of $n$ unsorted elements. Since \wref{inv:credits} is satisfied prior to the query, $4n$ credits are released. Quantity $M$ does not change. Thus, $4n$ units of potential are released, giving amortized time $n - 4n = -3n$. All intervals in the data structure resulting from the query are category $A$ intervals. The analysis of the preceding section for category $A$ intervals applies. We can pay $O(x)$ to satisfy \wref{inv:credits} on the smaller gap, and the remaining $3n$ units of released potential are enough to guarantee \wref{inv:credits} holds on all intervals in the larger gap.

Now suppose $\Delta_i$ is $1$-sided. If the query rank $r$ is closer to the side of $\Delta_i$ on which queries have been performed, then the same analysis of the preceding sections suffices. Note that there will be neither category $C$ nor category $D$ intervals. The creation of $2$-sided gap $\Delta'_i$ out of elements of $1$-sided gap $\Delta_i$ will cause $10x$ additional units of potential to be released due to the decrease in $M$; these units are not used in this case.

We are left with the case $\Delta_i$ is $1$-sided and the query rank $r$ is closer to the side of $\Delta_i$ on which queries have not been performed; suppose without loss of generality that previously only the right endpoint
of $\Delta_i$ has been queried and $r$ is closer to the left endpoint.
In this case, the creation of $2$-sided gap $\Delta'_{i+1}$ out of elements of $1$-sided gap $\Delta_i$ will cause $10cx$ units of potential to be released due to the decrease in $M$. Since $c \geq 1$, this is at least $5|\Delta_i|$ units of potential. We use them as follows. Answering the query takes no more than $O(|\Delta_i|)$ time, and ensuring intervals satisfy \wref{inv:credits} in new gaps $\Delta'_i$ and $\Delta'_{i+1}$ after the query similarly takes no more than $|\Delta_i|$ credits, which costs $4|\Delta_i|$ units of potential. Thus, in total this takes no more than $|\Delta_i| + 4|\Delta_i| - 5|\Delta_i| = O(1)$ amortized time.

\bigskip\noindent
Putting the preceding three sections together, we may answer a query in $O(\log n + x \log c)$ time while ensuring \wref{inv:credits} for future operations.

\subsection{Deletion}

The analysis of deletion of $e=(k,v)$ pointed to by \texttt{ptr} is as follows. The element $e$ can be removed from the interval in which it resides in $O(1)$ time. Removing said interval lazily, if applicable, takes $O(1)$ time. If the gap in which $e$ resides also needs removal, \wref{lem:gapds} says doing so will take $O(\log n)$ time.

In any case, when element $e \in \Delta_i$ is deleted, we must ensure \wref{inv:credits} on the remaining intervals of $\Delta_i$. If $e$ was outside of an interval $\mathcal{I}_{i,j}$, $o(\mathcal{I}_{i,j})$ decreases by one. Thus, for any such intervals, we pay one credit to ensure \wref{inv:credits} remains satisfied. Thus in accordance with \wref{lem:numintervals}, this takes $O(\log |\Delta_i|)$ total credits.

The total amortized cost is thus no more than $O(\log n + \log |\Delta_i|) = O(\log n)$. If the data structure of \wref{lem:gapds} supports operations in worst-case time, this runtime is also worst-case.

\subsection{Change-Key}
\label{sec:changekeyanalysis}

We analyze the change-key operation as follows. Suppose \texttt{ptr} points to element $e=(k,v)$ and we change its key as described in \wref{sec:changekeyalg} to $k'$. If $k'$ falls outside gap $\Delta_i$, $O(\log n)$ complexity follows from deletion and re-insertion of $(k',v)$. Otherwise, the binary search in $\Delta_i$ takes $O(\log \log |\Delta_i|)$ time, again by \wref{lem:numintervals}. To ensure \wref{inv:credits} on the intervals of $\Delta_i$, as is the case for deletion, we must pay one credit per interval $e$ is no longer outside of. Thus, the key-change operation takes at most $O(\log |\Delta_i|)$ time; however, if we change the key of $e$ towards the nearest query rank, we can show \wref{inv:credits} is satisfied without spending any credits.

At any point in time, all intervals in $\Delta_i$ are classified as being on the left side or the right side according to the closest query rank, in accordance to \wref{rule:left-right}. Any element of a left-side interval can have its key decreased, while only increasing or keeping constant the number of elements outside of any other left-side interval. The same is true for key increases of elements in right side intervals. 

Now consider if $e \in \mathcal{I}_{i,j}$ and $\mathcal{I}_{i,j}$ is the rightmost interval on the left side. Then we can also increase the key of $e$ while keeping the same or increasing the number of elements outside of any interval in $\Delta_i$. The same is true of decreasing the key of an element in the leftmost interval on the right side. Since the median of $\Delta_i$ falls in either the leftmost interval of the right side or the rightmost interval of the left side, it follows that we can ensure \wref{inv:credits} as long as the element whose key changes moves closer to its nearest query rank. Note that this analysis holds even as intervals change side designations due to insertions; for a refresher of this analysis see the proof of \wref{lem:ruleAinvariant}. This is despite delaying the application of \wref{rule:merge} until the following query in gap~$\Delta_i$.

This proves our statement in \wref{thm:main} about change-key.
The dichotomy displayed therein between cheap and expensive key changes can be refined as follows.
Suppose $c \geq 2$ is such that $e$ is located between (gap-local) ranks $|\Delta_i|/c$ and $|\Delta_i|-|\Delta_i|/c$ in $\Delta_i$; then we can change its key \emph{arbitrarily} in $O(\log \log \Delta_i + \log c)$ time. 
This is because of the geometric nature of interval sizes. Intervals are highly concentrated close to the edges of gap $\Delta_i$ in order to support queries that increase $B$ very little, efficiently. Thus, we can support arbitrary key changes in $O(\log \log |\Delta_i|)$ time for the vast majority of the elements of gap $\Delta_i$, since ensuring \wref{inv:credits} will only require a constant number of credits,
and the performance smoothly degrades as the changed elements get closer to previous query ranks.

A second refinement is that we can change $e$ arbitrarily without paying any credits if an insertion closer to the endpoint of gap $\Delta_i$ has happened before said key-change, but after the query that created $\Delta_i$: such insertion increases the number of elements outside of all intervals that are potentially affected by moving $e$ closer to the middle of $\Delta_i$, thus no credits have to be supplied.
A similar argument shows that the time complexity of deletion is only $O(1)$ if an element was previously inserted closer to the gap endpoint than the deleted element.
We point out again that, from the perspective of the data structure, these savings are realized automatically and the data structure will always run as efficiently as possible;
the credits are only an aspect of the analysis, not our algorithms.

\bigskip\noindent
In the following section, we show that a bound on the number of created intervals can bound the number of pointers required of the data structure and the insertion and change-key complexities when the number of queries is small.

\subsection{Pointer Bound and Improved Insertion and Change-Key}
\label{sec:qbounds}

The preceding sections show insertion into gap $\Delta_i$ in $O(\log (n/|\Delta_i|) + \log \log |\Delta_i|)$ time and a change-key time complexity of $O(\log \log |\Delta_i|)$. A bound of $O(\log q)$ can also be made, which may be more efficient when $q$ is small. We also prove the bound stated in \wref{thm:qbounds} on the total number of pointers required of the data structure. We address the latter first.

\begin{proof}[Proof of \wref{thm:qbounds}{}]
	Each query (including \texttt{Split($r$)} queries) creates at most $6$ intervals, and no other operations create intervals. The number of pointers required of all interval data structures is linear in the number of total intervals created, bounded to at most $n$. This is because elements within an interval are contiguous (in the sense an expandable array is contiguous) unless the interval is a result of merged intervals, where we assume that intervals are implemented as linked lists of arrays. Each merged interval must have been created at some point in time, thus the bound holds. The number of pointers required in the data structure of \wref{lem:gapds} is linear in the number of gaps (or intervals, if the data structure operates directly over intervals), taking no more than $O(\min(q, n))$ pointers, as the number of intervals is $O(\min(q, n))$.
\end{proof}

The above proof shows that the number of intervals and gaps in the entire data structure can be bounded by $q$. This implies the binary searches during insertion (both in the data structure of \wref{lem:gapds} and in \wref{sec:insertalg}) and change-key operations take no more than $O(\log q)$ time. 
This gives a refined insertion time bound of $O(\min(\log (n/|\Delta_i|) + \log \log |\Delta_i|, \log q))$ and a change-key time bound of $O(\min(\log q, \log \log |\Delta_i|))$.
To guarantee an $O(\log q$) time bound in the gap data structure, we can maintain
all gaps additionally in a standard balanced BST, with pointers between corresponding nodes
in both data structures. A query can alternatively advance from the root in both structures,
succeeding as soon as one search terminates.
Updates must be done on both structures, but the claimed $O(\log n)$ time bounds (for queries, delete, split, and merge) permit this behavior.

\section{Bulk Update Operations}
\label{sec:bulk}

Lazy search trees can support binary search tree bulk-update operations. We can split a lazy search tree at a rank $r$ into two lazy search trees $T_1$ and $T_2$ of $r$ and $n-r$ elements, respectively, such that for all $x \in T_1$, $y \in T_2$, $x \leq y$. We can also support a merge of two lazy search trees $T_1$ and $T_2$ given that for all $x \in T_1$, $y \in T_2$, $x \leq y$.

We state this formally in \wref{lem:bulk}.

\begin{lemma}
\label{lem:bulk}
Operation \texttt{Split($r$)} can be performed on a lazy search tree in time the same as \texttt{RankBasedQuery($r$)}. Operation \texttt{Merge($T_1$,\,$T_2$)} can be performed on lazy search trees in $O(\log n)$ worst-case time.
\end{lemma}
\begin{proof}
To perform operation \texttt{Split($r$)}, we first query for rank $r$ in the lazy search tree. We then split the data structure of \wref{lem:gapds} at the separation in gaps induced by the query for rank $r$. Two lazy search trees result, with their own future per-operation costs according to the number of elements and gaps that fall into each tree. Using a globally-biased $2, b$ tree~\cite{Bent85} with weights as in the proof of \wref{lem:gapds}, the split takes $O(\log n)$ worst-case time (Theorem 10 of~\cite{Bent85}). The overall time complexity is dominated by the query for rank $r$ in the original tree, since queries take at least $\Omega(\log n)$ time.

To perform operation \texttt{Merge($T_1$,\,$T_2$)}, we perform a merge on the data structures of \wref{lem:gapds} associated with each lazy search tree. Future per-operation costs are adjusted according to the union of all gaps and totaling of elements in the two lazy search trees that are combined. Using a globally-biased $2, b$ tree~\cite{Bent85} with weights as in the proof of \wref{lem:gapds}, the merge takes $O(\log n)$ worst-case time or better (Theorem 8 of~\cite{Bent85}).
\end{proof}

\wref{lem:bulk} completes the analysis for the final operations given in \wref{thm:main}.

\section{Average Case Insertion and Change-Key}
\label{sec:average}

Our time bounds from \wref{thm:main} are an additive $O(\log \log n)$ away from the optimal
time of insertion and change-key;
it turns out that in certain average-case scenarios, we can indeed reduce this time
to an optimal \emph{expected} amortized time.
The essential step will be to refine the binary search within a gap to an exponential search.

\subsection{Insert}

Recall that we store intervals in a sorted array.
We modify the insertion algorithm of the interval data structure in \wref{sec:insertalg} 
so that we instead perform a \emph{double binary search} 
(also called \emph{exponential search}~\cite{Bentley76}), 
outward from the last interval on the left side and first interval on the right side. 
This is enough to prove the following result.

\begin{theorem}[Average-case insert]
\label{thm:average-insert}
	Suppose the intervals within a gap are balanced using \wref{rule:merge} and
	further suppose insertions follow a distribution such that the gap in which an inserted element 
	falls can be chosen adversarially, but amongst the elements of that gap, 
	its rank is chosen uniformly at random. 
	Then insertion into gap $\Delta_i$ takes expected time $O(\log(n/|\Delta_i|))$.
\end{theorem}
\begin{proof}
First note that the double binary search during insertion finds an interval that is 
$k$ intervals from the middlemost intervals in time $O(\log k)$;
apart from constant factors, this is never worse than the $O(\log \ell_i)$ of a binary search.

The assumption on insertion ranks implies that the probability to insert into interval $\mathcal I_{i,j}$
(out of the possible $\ell_i$ intervals in gap $\Delta_i$) is $|\mathcal I_{i,j}|/ |\Delta_i| \pm O(1/|\Delta_i|)$,
\ie, proportional to its size.
Recall that in a gap $\Delta_i$ satisfying \wref{lem:numintervals}, interval sizes grow at least
like $(\sqrt2)^k$; that implies the largest (middlemost) intervals contain 
a constant fraction of the elements in $\Delta_i$; 
for these, insertion takes $O(1)$ time.
The same applies recursively:
With every outward step taken, the insertion procedure takes $O(1)$ more time, 
while the number of elements that fall in these intervals decreases by a constant factor. 
The expected insertion time in the interval data structure is proportional to
\[
	\sum_{k=1}^\infty \frac{\log k}{(\sqrt 2)^k} \;\leq\; 
	\sum_{k=1}^\infty \frac{k}{(\sqrt 2)^k} \;=\; 
	4 + 3 \sqrt 2,
\]
\ie, constant overall.
Adding the $O(\log(n/|\Delta_i|))$ time to find the gap yields the claim.
\end{proof}

Observe that walking from the largest intervals outward, instead of performing an exponential search~\cite{Bentley76}, is sufficient for the above analysis. However, the exponential search also satisfies the worst case $O(\log \log n)$ bound (more precisely $O(\min(\log \log |\Delta_i|, \log q))$) described in \wref[Sections]{sec:insertalg} and~\ref{sec:insertanalysis}.

\begin{remark}[Fast insertion without arrays]
We can achieve the same effect if intervals are stored in another biased search tree so that interval $\mathcal{I}_{i,j}$ receives weight $|\Delta_i|/\ell_i + |\mathcal{I}_{i,j}|$.
\end{remark}

\wref{thm:average-insert} assumes that intervals are balanced according to \wref{rule:merge}. 
In \wref{sec:ds}, we described balancing according to \wref{rule:merge} lazily. 
Keeping \ref{rule:merge} balanced while insertions or change-key operations occur, 
in the required time complexity, is nontrivial. 
We show it can be done in $O(1)$ amortized time below.

\begin{lemma}[Strict merging]
\label{lem:Bcheck}
Given a gap $\Delta_i$, we can keep intervals in $\Delta_i$ balanced according 
to within a constant factor of the guarantee of \wref{rule:merge} in $O(1)$ 
amortized time per insertion into $\Delta_i$.
\end{lemma}
\begin{proof}
We utilize the exponentially-increasing interval sizes due to \wref{lem:numintervals}. We check the outermost intervals about every operation and exponentially decrease checking frequency as we move inwards. The number of intervals checked over $k$ operations is $O(k)$. The guarantee of \wref{rule:merge} is changed so that the number of elements left of $\mathcal{I}_{i,j}$ in $\Delta_i$ is no more than a constant times $|\mathcal{I}_{i,j}| + |\mathcal{I}_{i,j+1}|$ (reflected for right side intervals), to which previous analysis holds.
\end{proof}

\subsection{Change-Key}

If we apply \wref{lem:Bcheck}, we can also support improved average-case change-key operations
in the following sense.

\begin{theorem}[Average-case change-key]
\label{thm:average-changekey}
If a \textup{\texttt{ChangeKey(ptr,\,$k'$)}} operation is performed such that the element pointed to by \texttt{ptr}, $e=(k,v)$, moves closer to its closest query rank within its gap and the rank of $k'$ is selected uniformly at random from valid ranks, it can be supported in $O(1)$ expected time.
\end{theorem}
\begin{proof}
We again perform a double binary search (exponential search~\cite{Bentley76}) for the new interval of $e$;
this time we start at the interval $\mathcal{I}_{i,j}$ in which $e$ currently resides and move outwards from there. 
The analysis follows similarly to \wref{thm:average-insert}.
\end{proof}

When used as a priority queue, \wref{thm:average-changekey} 
improves the average-case complexity of decrease-key to $O(1)$.

\section{Randomized-Selection Variant}
\label{sec:random}

We can improve the practical efficiency of lazy search trees by replacing exact median-finding in the query procedure with randomized pivoting. Specifically, after finding sets $L$ and $R$ as described in \wref{sec:queryalg}, we then partition $L$ into sets $L_l$ and $L_r$ by picking a random element $p \in L$ and pivoting so that all elements less than $p$ are placed in set $L_l$ and all elements greater than $p$ are placed in set $L_r$. To avoid biasing when elements are not unique, elements equal to $p$, should be split between $L_l$ or $L_r$. 
We then repeat the procedure one more time on set $L_r$. We do the same, reflected, for set $R$.

\begin{remark}[Partitioning with equal keys]
	In our analysis, we assume for simplicity that the number of 
	elements with same key as $p$, including $p$ itself, 
	that are assigned to the left segment is chosen uniformly at random from the number of copies.
	That implies overall a uniform distribution for the size of the segments.
	Partitioning procedures as used in standard implementations of quicksort~\cite{Sedgewick1978} 
	actually lead to slightly more balanced splits~\cite{Sedgewick1977a}; they will only perform better.
	For practical implementations of lazy search trees, choosing the partitioning element $p$
	as the median of a small sample is likely to improve overall performance.
\end{remark}

Changing the query algorithm in this way requires a few changes in our analysis. The analysis given in \wref{sec:analysis} is amenable to changes in constant factors in several locations. Let us generalize the potential function as follows, where $\alpha$ is a set constant, such as $\alpha = 4$ in \wref{sec:analysis}. One can see from \wref{sec:onesided} that this will imply the constant in front of $M$ must be at least $2(\alpha+1)$.
\[
\Phi \;\;=\;\; 2(\alpha+1)M \;+\; \alpha\sum_{\mathclap{\substack{1 \leq i \leq m,\\1 \leq j \leq \ell_i}}} \: c(\mathcal{I}_{i,j}).
\]

Insertion still takes $O(\min(\log(n/|\Delta_i|) + \log \log |\Delta_i|, \log q))$ time. As before, splitting into sets $L$ and $R$ can typically be done in $O(|\mathcal{I}_{i,j}|)$ deterministic time via the result of the query, but if not, quickselect can be used for $O(|\mathcal{I}_{i,j}|)$ expected (indeed with high probability) time performance~\cite{Hoare61,FloydRivest1975,Kiwiel2005}. The modified pivoting procedure described above for $L_l$ and $L_r$ is repeated in total 4 times. We can thus bound the complexity of these selections at $O(|\mathcal{I}_{i,j}|)$, regardless of the randomization used.

Then by application of \wref{lem:cancel}, we reduce the current amortized time to split $\mathcal{I}_{i,j}$ to $O(x)$, leaving $(\alpha-1)c(\mathcal{I}_{i,j})$ units of potential to handle ensuring \wref{inv:credits} on category A intervals in \wref{sec:ensure}.

The number of credits necessary to satisfy \wref{inv:credits} on category $A$ intervals is now a random variable. %
Recall the arguments given in \wref{sec:current} and \wref{sec:ensure} regarding category $A$ intervals. As long as the (expected) number of credits to satisfy \wref{inv:credits} on category $A$ intervals is at most a constant fraction $\gamma$ of $|\mathcal{I}_{i,j}|$, we can set $\alpha = \frac{1}{1-\gamma}$ and the amortized analysis carries through. 

We have the following regarding the expected number of credits to satisfy \wref{inv:credits} on category $A$ intervals using the randomized splitting algorithm.

\begin{lemma}
\label{lem:expectedcost}
	Suppose a query falls in interval $\mathcal{I}_{i,j}$ and the intervals built from the elements of $\mathcal{I}_{i,j}$ are constructed using the randomized splitting algorithm. The expected number of credits necessary to satisfy \wref{inv:credits} on category $A$ intervals after a query is no more than $\frac{143}{144}|\mathcal{I}_{i,j}| + O(1)$.
\end{lemma}
\begin{proof}
	We prove the loose bound considering only one random event in which a constant fraction of $|\mathcal{I}_{i,j}|$ credits are necessary, which happens with constant probability.
	
	We orient as in \wref{sec:ensure}, assuming the larger new gap, $\Delta'_{i+1}$, is right of the smaller new gap, $\Delta'_i$. We must consider the number of credits necessary to satisfy \wref{inv:credits} on the three category $A$ intervals $\mathcal{I'}_{i+1,1}$, $\mathcal{I'}_{i+1,2}$, and $\mathcal{I'}_{i+1,3}$ of new gap $\Delta'_{i+1}$. The rightmost interval $\mathcal{I'}_{i+1,3}$ has size drawn uniformly at random in $1, \ldots, |R|$, the leftmost, $\mathcal{I'}_{i+1,1}$, takes size uniformly at random from the remaining elements, and the middlemost interval $\mathcal{I'}_{i+1,2}$ takes whatever elements remain.
	
	Suppose the rightmost interval $\mathcal{I'}_{i+1,3}$ comprises a fraction of
	$x = |\mathcal I'_{i+1,3}| / |R| \in \bigl[\frac{1}{3},\,\frac{2}{3}\bigr]$ of all elements in $R$, and further suppose the leftmost interval $\mathcal{I'}_{i+1,1}$ takes between $1/2$ and $3/4$ of the remaining elements, \ie,
	a fraction $y = |\mathcal I'_{i+1,1}|/|R| \in \bigl[\frac{1}{2}(1-x),\,\frac{3}{4}(1-x)\bigr]$ of the overall elements in $R$. In this case, it is guaranteed we require no credits to satisfy \wref{inv:credits} on the middlemost interval. The number of credits to satisfy \wref{inv:credits} on the rightmost and leftmost intervals is $(x+y)|R|$, which is maximized at $\frac{11}{12}|R|$. This event happens with probability $\frac{1}{3} \cdot \frac{1}{4} - O\left({1}/{|R|}\right) = \frac{1}{12} - O\left({1}/{|R|}\right) $, where we include the $O\left({1}/{|R|}\right)$ term to handle rounding issues with integer values of $|R|$. As we never require more than $|R|$ credits in any situation and $|R|\le |\mathcal I_{i,j}|$, we can then bound the expected number of necessary credits at $\frac{11}{12}\cdot |\mathcal{I}_{i,j}| + \frac{1}{12} \cdot \frac{11}{12}|\mathcal{I}_{i,j}| + O(1) = \frac{143}{144}|\mathcal{I}_{i,j}| + O(1)$.
\end{proof}

With \wref{lem:expectedcost}, we can set $\alpha = 144$ and use the remaining $143c(\mathcal{I}_{i,j})$ credits from destroying $\mathcal{I}_{i,j}$ and bound $143|\mathcal{I}_{i,j}|-143c(\mathcal{I}_{i,j})$ with $143x$ via \wref{lem:cancel}. All other query analysis in \wref{sec:ensure} is exactly as before. This gives total expected amortized query time $O(\log n + x \log c)$ on $2$-sided gaps. With a constant of $2(\alpha+1)$ in front of $M$ in the generalized potential function, the analysis for $0$ and $1$-sided gaps in \wref{sec:onesided} carries through.

Putting it all together, we get the following result.

\begin{theorem}[Randomized splitting]
\label{thm:expected}
If partitioning by median in the query algorithm is replaced with splitting on random pivots, lazy search trees satisfy the same time bounds, in worst-case time, as in \wref{thm:main}, except that \texttt{RankBasedQuery($r$)} and \texttt{Split($r$)} now take $O(\log n + x \log c)$ expected amortized time.
\end{theorem}

Note that another possible approach is to change \wref{inv:credits} to something like $c(\mathcal{I}_{i,j}) + 2o(\mathcal{I}_{i,j}) \geq |\mathcal{I}_{i,j}|$, which gives further flexibility in the rest of the analysis. This is, however, not necessary to prove \wref{thm:expected}.

\section{Lazy Splay Trees}
\label{sec:splay}

Splay trees~\cite{Sleator85} are arguably the most suitable choice of a biased search tree in practice;
we thereby explore their use within lazy search trees in this section.
We show that an amortized-runtime version of \wref{lem:gapds} can indeed be obtained using splay trees.
We also show that by using a splay tree, the efficient access theorems of the splay tree are achieved automatically
by the lazy search tree. This generalizes to any biased search tree that is used as the data structure of \wref{lem:gapds}.

\subsection{Splay Trees For The Gap Data Structure}

We show that splay trees can be used as the gap data structure.

\begin{lemma}[Splay for Gaps]
\label{lem:gapds-amortized}
	Using \textit{splay trees} as the data structure for the set of gaps $\{\Delta_i\}$ allows support of all operations listed in \wref{lem:gapds}, where the time bounds are satisfied as \textit{amortized} runtimes over the whole sequence of operations.
\end{lemma}

\begin{proof}
We use a splay tree~\cite{Sleator85} and weigh gap $\Delta_i$ with $w_i = |\Delta_i|$. 
The sum of weights, $W$, is thus equal to $n$. 
Operation 1 can be supported by searching with $e=(k,v)$ into the tree 
until gap $\Delta_i$ is found and then splayed. 
According to the Access Lemma~\cite{Sleator85}, this is supported in 
$O(\log(n/|\Delta_i|))$ amortized time. 
Operation 2 requires a weight change on gap $\Delta_i$. By first accessing gap $\Delta_i$, 
so that it is at the root, and then applying a weight change, 
this operation can be completed in time proportional to the access. 
According to the Access Lemma~\cite{Sleator85} and the Update Lemma~\cite{Sleator85}, 
this will then take $O(\log(n/|\Delta_i|))$ amortized time. 
Note that for our use of operation 2, the element will already have just been accessed, 
so the additional access is redundant. 
Operations 3 and 4 are supported in $O(\log n)$ time by the Update Lemma~\cite{Sleator85}.
Note that when the gap data structure is used in a lazy search tree, it always starts empty
and more gaps are added one by one when answering queries. Hence any sequence of operations
arising in our application will access every element in the splay tree at least once.
\end{proof}

Note that a bound of $O(\log q)$ amortized cost for all operations also holds by using 
equal weights in the analysis above (recall that in splay trees, the node weights are solely a means
for analysis and do not change the data structure itself). 

\subsection{Efficient Access Theorems}
\label{sec:efficient-access}

We now specify a few implementation details to show how lazy search trees can perform accesses as fast as the data structure of \wref{lem:gapds} (resp.\ \wref{lem:gapds-amortized}).%

If an element $e$ is the result of a query for a second time, then during that second access, $e$ is the largest element in its gap.
Instead of destroying that gap, we can assume the identity of the gap $e$ falls into after the query to be the same gap in which $e$ previously resided (depending on implementation, this may require a key change in the data structure of \wref{lem:gapds}, but the relative ordering of keys does not change). In this way, repeated accesses to elements directly correspond to repeated accesses to nodes in the data structure of \wref{lem:gapds}. 
Further, implementation details should ensure that no restructuring occurs in the interval data structure when an element previously accessed is accessed again. This is implied by the algorithms in \wref{sec:ds}, but care must be taken in the case of duplicate elements. This will ensure accessing a previously-accessed element will take $O(1)$ time in the interval data structure.

With these modifications, the lazy search tree assumes the efficient access properties of the data structure of \wref{lem:gapds}. We can state this formally as follows.

\begin{theorem}[Access Theorem]
	\label{thm:efficientaccess}
	Given a sequence of element accesses, lazy search trees perform the access sequence in time no more than an additive linear term from the data structure of \wref{lem:gapds}, disregarding the time to access each element for the first time.%
\end{theorem}
\begin{proof}
Once every item has been accessed at least once, the data structures are the same, save for an extra $O(1)$ time per access in the interval data structure. The cost of the first access may be larger in lazy search trees due to necessary restructuring.%
\end{proof}

While we would ideally like to say that lazy search trees perform within a constant factor of splay trees on any operation sequence, this is not necessarily achieved with the data structure as described here. Time to order elements on insertion is delayed until queries, implying on most operation sequences, and certainly in the worst case, that lazy search trees will perform within a constant factor of splay trees, often outperforming them by more than a constant factor. However, if, say, elements $1, 2, \ldots, n$ are inserted in order in a splay tree, then accessed in order $n, n-1, \ldots, 1$, splay trees perform the operation sequence in $O(n)$ time, whereas lazy search trees as currently described will perform the operation sequence in $O(n \log n)$ time.

\wref{thm:efficientaccess} shows using a splay tree for the gap data structure (\wref{lem:gapds-amortized}) allows lazy search trees to achieve its efficient-access theorems. Observing that the initial costs of first access to elements total $O(n \log n)$, we achieve \wref{cor:splayaccess} below.

\begin{corollary}
\label{cor:splayaccess}
Suppose a splay tree is used as the gap data structure. Then lazy search trees achieve the efficient access theorems of the splay tree, including static optimality, static finger, dynamic finger, working set, scanning theorem, and the dynamic optimality conjecture~\cite{Sleator85,Cole2000a,Cole2000b,Elmasry04}.
\end{corollary}

\section{Conclusion and Open Problems}
\label{sec:conclude}

We have discussed a data structure that improves the insertion time of binary search trees, when possible. Our data structure generalizes the theories of efficient priority queues and binary search trees, providing powerful operations from both classes of data structures. As either a binary search tree or a priority queue, lazy search trees are competitive. From a theoretical perspective, our work opens the door to a new theory of insert-efficient order-based data structures.

This theory is not complete. Our runtime can be as much as an additive $O(n \log \log n)$ term from optimality in the model we study, providing $O(\log \log n)$ time insert and decrease-key operations as a priority queue when $O(1)$ has been shown to be possible~\cite{Fredman87}. Further room for improvement is seen in our model itself, where 
delaying insertion work further can yield improved runtimes on some operation sequences. We see several enticing research directions around improving these shortcomings and extending our work. We list them as follows:

\begin{enumerate}
	\item Extend our model and provide a data structure so that the order of operations performed is significant. A stronger model would ensure that the number of comparisons performed on an inserted element depends only on the queries performed after that element is inserted.
	\item Within the model we study, improve the additive $O(n \log \log n)$ term in our analysis to worst-case $O(n)$, or give a lower bound that shows this is not possible while supporting all the operations we consider.
	\item Explore and evaluate competitive implementations of lazy search trees. In the priority queue setting, evaluations should be completed against practically efficient priority queues such as binary heaps~\cite{Williams64}, Fibonacci heaps~\cite{Fredman87}, and pairing heaps~\cite{Fredman86}. On binary search tree workloads with infrequent or non-uniformly distributed queries, evaluations should be completed against red-black trees~\cite{Bayer72}, AVL trees~\cite{Adelson62}, and splay trees~\cite{Sleator85}.
	\item Support efficient general merging of unordered data. Specifically, it may be possible to support $O(1)$ or $O(\log n)$ time merge of two lazy search trees when both are used as either min or max heaps.
	\item Although the complexity of a rank-based query must be $\Omega(n)$ when the query falls in a gap of size $|\Delta_i| = \Omega(n)$, the per-operation complexity of \texttt{RankBasedQuery($r$)} could potentially be improved to $O(x \log c + \log n)$ worst-case time instead of amortized time, with $x$ and $c$ defined as in \wref{thm:main}.
	\item Develop an external memory version of lazy search trees for the application of replacing B-trees~\cite{Bayer72}, $B^\epsilon$ trees~\cite{Brodal03}, or log-structured merge trees~\cite{ONeil96} in a database system.
	\item Investigate multidimensional geometric data structures based off lazy search trees. Range trees~\cite{Bentley79}, segment trees~\cite{Bentley77}, interval trees~\cite{Edelsbrunner80,McCreight80}, kd-trees~\cite{Bentley75}, and priority search trees~\cite{McCreight85} are all based on binary search trees. By building them off lazy search trees, more efficient runtimes as well as space complexity may be possible.
\end{enumerate}

Regarding point 3, we have implemented a proof-of-concept version of a lazy search tree in C$++$, taking no effort to optimize the data structure. Our implementation is roughly 400 lines of code not including the gap data structure, to which we use a splay tree~\cite{Sleator85}. Intervals are split via randomized pivoting, as described in \wref{sec:random}. The optimization to support $O(1)$ time average case insertion into the interval data structure is implemented, and the data structure also satisfies the $O(\min(q, n))$ pointer bound by representing data within intervals in a linked list of C$++$ vectors.

Our implementation has high constant factors for both insertion and query operations. For insertion, this is likely due to several levels of indirection, going from a gap, to an interval, to a linked list node, to a dynamically-sized vector. For query, this is likely due to poor memory management. Instead of utilizing swaps, as in competitive quicksort routines, our implementation currently emplaces into the back of C$++$ vectors, a much slower operation. The current method of merging also suggests some query work may be repeated, which although we have shown does not affect theoretical analysis, may have an effect in practice.

Still, initial experiments are promising. Our implementation outperforms both the splay tree which our implementation uses internally as well as C$++$ \texttt{set}, for both low query load scenarios and clustered queries. To give a couple data points, on our hardware, with $n = 1\,000\,000$, our implementation shaves about 30\% off the runtime of the splay tree when no queries are performed and remains faster for anything less than about $2\,500$ uniformly distributed queries. When $n=10\,000\,000$, our implementation shaves about 60\% off the runtime of the splay tree when no queries are performed and remains faster for anything less than about $20\,000$ uniformly distributed queries. The C$++$ \texttt{set} has runtime about 30\% less than our splay tree on uniformly distributed query scenarios. Our experiments against C$++$ STL \texttt{priority\_queue} show that our current implementation is not competitive.

Finally, regarding points 2 and 4, since this article was written we have succeeded in devising a solution using very different techniques that removes the $O(\log \log n)$ terms and supports constant time priority queue merge. The new solution requires sophisticated data structures and is not based on arrays, so the approach discussed herein is likely to be more practical.



%

%

%

%

\subsection*{Acknowledgements}

The authors of this paper would like to thank Ian Munro, Kevin Wu, Lingyi Zhang, Yakov Nekrich, and Meng He for useful discussions on the topics of this paper. We would also like to thank the anonymous reviewers for their helpful suggestions.

\bibliographystyle{alphaurl}
\let\oldthebibliography\thebibliography
\renewcommand\thebibliography[1]{%
	\oldthebibliography{#1}%
	\pdfbookmark[1]{References}{}%
}

\bibliography{ref}

\end{document}